\documentclass[12pt]{article}
\usepackage[left=1in,top=0.8in,right=1in,bottom=0.8in,head=.2in]{geometry}
\usepackage{graphicx} 

\usepackage{mystyle}

\usepackage{todonotes}

\newcommand{\full}{\mathrm{full}}
\newcommand{\sing}{\mathrm{sgtn}}
\newcommand{\row}{\mathrm{row}}
\newcommand{\mis}{\mathrm{mis}}

\newcommand{\footremember}[2]{%
    \footnote{#2}
    \newcounter{#1}
    \setcounter{#1}{\value{footnote}}%
}

\begin{document}

\def\spacingset#1{\renewcommand{\baselinestretch}%
{#1}\small\normalsize} \spacingset{1}


\newcommand{\titletext}{Network-based Neighborhood Regression}
\newcommand{\removelinebreaks}[1]{%
      \def\\{\relax}#1}
\def\titleRLB{\removelinebreaks{\titletext}}

\title{\titletext}

\date{}

\bigskip
\bigskip
\author{Yaoming Zhen\footremember{uto}{Department of Statistical Sciences, University of Toronto, Toronto, Ontario, M5G 1X6, Canada}
\and Jin-Hong Du\footremember{cmustats}{Department of Statistics and Data Science, Carnegie Mellon University, Pittsburgh, PA 15213, USA.}\footremember{cmumld}{Machine Learning Department, Carnegie Mellon University, Pittsburgh, PA 15213, USA.}}

\maketitle

\begin{abstract}
    Given the ubiquity of modularity in biological systems, module-level regulation analysis is vital for understanding biological systems across various levels and their dynamics.
    Current statistical analysis on biological modules predominantly focuses on either detecting the functional modules in biological networks or sub-group regression on the biological features without using the network data.
    This paper proposes a novel network-based neighborhood regression framework whose regression functions depend on both the \emph{global} community-level information and \emph{local} connectivity structures among entities.
    An efficient community-wise least square optimization approach is developed to uncover the strength of regulation among the network modules while enabling asymptotic inference.
    With random graph theory, we derive non-asymptotic estimation error bounds for the proposed estimator, achieving exact minimax optimality. 
    Unlike the root-$n$ consistency typical in canonical linear regression, our model exhibits linear consistency in the number of nodes $n$, highlighting the advantage of incorporating neighborhood information.
    The effectiveness of the proposed framework is further supported by extensive numerical experiments.
    Application to whole-exome sequencing and RNA-sequencing Autism datasets demonstrates the usage of the proposed method in identifying the association between the gene modules of genetic variations and the gene modules of genomic differential expressions.
\end{abstract}

\noindent%

\spacingset{1.8}
\begin{bibunit}[apalike]

\section{Introduction}

In various biological systems, it is more than common for biological units to interplay with each other and form functional modules, such as in the gene co-expression networks \citep{liu2015network}, protein-protein interaction networks \citep{brohee2006evaluation}, and functional connectivities in brain regions \citep{zhang2020mixed, chen2020random}. Measurements of a single biological unit depend not only on its own features but also on those of other units it interacts with. 
Understanding the module-level regulation relationships could provide crucial insights into the biological development processes. It is thus of scientific interest to investigate the evolution of the biological units while incorporating their local neighborhood information \citep{zhang2017estimating} and global cluster-level relationships \citep{le2022linear} into a unified framework.

Our motivation is an study of genetic and genomic associations related to autism spectrum disorder (ASD). 
Characterized by compromised social interactions and repetitive behaviors, ASD is significantly influenced by genetic variation, which is usually quantified as the genetic risk (GR) scores computed from the whole exome sequencing datasets \citep{liu2013analysis}.
Although there are typically thousands of genes, the genetic evidence indicated by the GR scores is scarce.
Using hidden Markov random field (HMRF) models, \citet{liu2014dawn, liu2015network} incorporated gene co-expression networks to identify clusters of autism risk genes.
Recently, \citet{gandal2022broad} study the genomic differentially expressed (DE) scores by contrasting the gene expressions between ASD and neurotypical individuals.
When comparing the GR scores in previous studies with DE scores, only a small portion of overlaps is observed, and the interplay between genetic evidence and genomic evidence remains unknown.
Further, the evolution of a gene's expression levels shall relate not only to its neighboring genes but also to the functional module in which it is located.
Therefore, a better modeling strategy is desired to quantify the directional causal effect from the genetic evidence to the genomics evidence while incorporating both the neighborhood and community information.

Existing methods, such as the random effects and subgroup effects models, only provide information about heterogeneity within each community but lack the capability to model inter-community interactions. Additionally, approaches like network-assisted regression proposed by \citet{li2019prediction,le2022linear} do not directly incorporate neighborhood information or account for the heterogeneity of regression coefficients across different communities. 
Although these methods attempt to leverage network data, they fall short of comprehensively modeling the complex dependencies that exist among samples and network modules.

To bridge these gaps, we propose a novel network-based neighborhood regression model that predicts the response of a node based on the covariates of all nodes within its neighborhood. 
The significant challenge lies in addressing the dependency among samples within the network, leading to potential overparameterization issues. To mitigate this, our model employs a block structure in the neighborhood regression coefficient matrix to reflect community-wise common effects, making estimation and inference feasible. 
We demonstrate that the community-wise least squares objective function can be decomposed into multiple non-overlapping linear regression objective functions, which allows for efficient estimation and inference despite the complexities posed by network data with community-wise interactions.

The main contribution of this paper is threefold: (1) Aiming to better understand the directional effect of the autism genetic factors on their differential expressions, the proposed network-based neighborhood regression framework incorporates not only the \emph{local} connectivity patterns of the genes but also their \emph{global} community-wise common effects. 
(2) Theoretically, we develop random-design and non-asymptotic analyses for the network-based design matrix to derive concentration behavior for the Hessian matrix of neighborhood regression, which further leads to the asymptotic consistency of the proposed community-wise least square estimator. 
Most importantly, our theory, along with minimax optimality, suggests the blessing of neighborhood information aggregation, yielding that the convergence rate of the neighborhood regression coefficients is almost linear in the number of nodes if the network is dense enough.
This finding substantially distinguishes from the root-$n$ consistency in canonical linear regression setups, highlighting the potential of leveraging neighborhood information in network-based regression models. We further study the influence of misspecified community memberships. 
(3) Simulation studies showcase the feasibility and necessity of the proposed method, and application to the Autism gene datasets identifies interpretable community-wise common effects among the genes under investigation.

\subsection{Other applications}

While this paper primarily addresses biomedical applications, other potential applications of the proposed network-based neighborhood regression model are also worth considering.
\begin{itemize}
\item In social networks, such as Twitter, a person's retweet behavior  depends not only on how actively she/he posts original tweets but also on how actively her/his friends post original tweets. 
People of the same age and occupation group tend to share similar tweets and exhibit similar retweet patterns.
\item In the city transportation network, the number of arrival rides at any station depends not only on the number of departure rides from this station but also on the number of departure rides from its neighboring stations.
Stations in urban communities or suburban communities have similar arrival and departure patterns.
\item In the world trading network, the among of goods that a country imports depends on the among of goods it exports and those of its neighboring countries export.
Trading patterns in developed countries or developing countries can be similar.
\end{itemize}
In these examples, predicting individual response based on the covariates in its neighborhood and community necessitates a new regression method.

\subsection{Related work}
Statistical analysis of biological data with module structure primarily focuses on sub-group identification, such as community detection and sub-group regression.
Common community detection approaches include likelihood-based approaches under stochastic block model \citep{AlainConsistency2012}, latent space model \citep{raftery2012fast, zhang2022joint}, and random dot graph model \citep{athreya2018statistical}, spectral clustering under stochastic block model and degree-corrected stochastic block model \citep{Jin2015fast, LeiConsistency2015}, and modularity maximization \citep{shang2013community}. For details, we refer interested readers to the comprehensive review papers \citet{Abbe2018recent} and \citet{gao2021minimax}. Besides, there has been a notable shift in research focusing on integrating network structure and node attributes to identify communities more accurately. Related works include \citet{newman2016structure, zhang2016community, yan2021covariate, xu2023covariate} and \citet{hu2023network}.

In the research line of sub-group regression, \citet{zhou2022subgroup} and \citet{wang2023high} propose sub-group regression models to analyze personal treatment effects and low-dimensional latent factors, respectively, without using network information. 
However, utilizing the network information in predictive models has not yet been well-studied.
Recently, \citet{li2019prediction} study linear regression with network cohesion regularizer on the individual node effects; \citet{le2022linear} further extend it to a semi-parametric regression model by incorporating network spectral information.
However, neither method directly incorporates node-wise neighborhood information and the heterogeneity of regression coefficients in different communities.

\subsection{Notations} 
Denote $[n] =\{1,\ldots,n\} $ and $\ind(A)\in\{0,1\}$ as the indicator function for any event $A$. 
Let $\zero_n,\one_n\in\RR^n$ be the vectors of all zeros and ones, respectively, and $\bI_n\in\{0, 1\}^{n\times n}$ the $n$th order identity matrix. For a vector $\bx\in\RR^n$, denote by $\|\bx\|_{p}$ its $l_p$-norm with $p\in\NN\cup\{\infty\}$. Conventionally, we write $\|\bx\|$ as the $l_2$-norm of $\bx$ without the subscript. In addition, $\diag(\bx)\in\RR^{n\times n}$ denotes the diagonal matrix whose diagonals are $x_1, \ldots, x_n$. For a matrix $\bA\in\RR^{m\times n}$, $\bA_{i,\cdot}\in\RR^{n}$ and $\bA_{\cdot,j}\in\RR^{m}$  respectively represent its $i$th row and $j$th column, and we denote $\bA^\dag$ as its Moore-Penrose pseudoinverse. Moreover, we denote $\lambda_k(\bA)$ as the $k$th largest eigen-value of a symmetric matrix $\bA$, and the smallest and largest eigenvalues are also denoted by $\lambda_{\min}(\bA)$ and $\lambda_{\max} (\bA)$, respectively. 
If $\bA$ is positive definite, we have $\lambda_{\max} (\bA) = \|\bA\|$, the spectral norm of $\bA$, while $\lambda_{\min}(\bA) = \|\bA^{-1}\|^{-1}$. The regular matrix product, Hadamard product, Kronecker product, and Khatri–Rao product (column-wise Kronecker product) between two matrices $\bA$ and $\bB$ are denoted by $\bA \bB$, $\bA * \bB$, $\bA \otimes \bB$, and $\bA \odot \bB$, respectively. For convenience, we place the lowest operation priority on Hadamard products among the above products, e.g., $\bA \bB * \bC$ =  $(\bA \bB) * \bC$. Suppose $\bA$ and $\bB$ are conformable symmetric matrices, we write $\bA \preceq \bB$ if $\bB - \bA$ is positive semi-definite. Finally, for two positive sequences $a_n$ and $b_n$, $a_n = O(b_n)$ implies there exists an absolute constant $C$ such that $a_n \le Cb_n$ for all $n$, $a_n = \Omega(b_n)$ means $b_n = O(a_n)$, and $a_n = o(b_n)$ stands for $\lim_{n\rightarrow \infty} a_n/b_n = 0$. For convenience, we denote $a_n \asymp b_n$ if $a_n = O(b_n)$ and $a_n = \Omega(b_n)$. The terms ``module'', ``cluster", and  ``community'' will be used interchangeably. 

\section{Network-based neighborhood regression}\label{sec:NNR}

\subsection{Genetic risk and differential expressed scores modeling}\label{sec:model}
    
    Consider two sources of evidence from statistical tests: the genetic risk (GR) score $\bx = (x_i)_{i \in [n]} \in \RR^n$ and the differentially expressed (DE) score $\by = (y_i)_{i \in [n]} \in \RR^n$ for $n$ genes. These scores can be derived from previous genetic studies \citep{fu2022rare} and genomic analyses \citep{gandal2022broad}. Additionally, a gene co-expression network, revealing bivariate dependencies between gene expression patterns and their corresponding sub-networks (modules) \citep{liu2015network}, often serves as auxiliary information. Let $\bA = (A_{i, j})_{i, j \in [n]} \in \{0, 1\}^{n \times n}$ denote the adjacency matrix of an undirected and unweighted gene co-expression network. To characterize the directional effect from genetic evidence to genomic evidence with network information, we consider the following network-based neighborhood regression model:
    \begin{align}
        y_i = \sum_{j\in N_i} \tilde{\beta}_{i,j}x_j + \epsilon_i, \text{ for }i\in [n] \label{equ: neighborhood_regression}, 
    \end{align}
    where $N_i =\{j\in [n]: A_{i, j} = 1\} $ is the neighborhood of gene $i\in [n]$, $\tilde{\beta}_{i, j}$ is the effect sizes from gene $j$ to gene $i$ for $i,j\in[n]$, and $\bepsilon=(\epsilon_i)_{i\in[n]}$ contains independent additive noises.
    We further denote the coefficient matrix by $\tilde{\bbeta} = (\tilde{\beta}_{i, j})_{i,j\in[n]} \in\RR^{n\times n}$.

    Model \eqref{equ: neighborhood_regression} formulates $y_i$ as a linear combination of its neighbors' covariate $x_j$'s, up to additive noise.
    Specifically, the DE score of gene $i$ is affected by the GR scores of its neighboring $j$'s who is connected to $i$, for $i, j\in [n]$. 
    Because $y_i$ is supposed to be affected by $x_i$, we assume $A_{i, i} \equiv 1$ and thus $i \in N_i$. 
    When the network information is not available, this model reduces to the simple linear regression model with $N_i = \{i\}$ and $\tilde{\beta}_{i, i} = \beta$ for some constant $\beta\in\RR$, and $i\in [n]$. 
    For simplicity, model \eqref{equ: neighborhood_regression} does not include an intercept term as one can always center the data prior to model fitting; see \Cref{app:intercept} for more details.

    Compared to high-dimensional regressions, estimating the coefficient matrix $\tilde{\bbeta}$ in model \eqref{equ: neighborhood_regression} is particularly challenging, even with the presence of only a \emph{single} covariate. The fundamental difficulty stems from the model's overparameterization: while the model includes $n^2$ unknown parameters, there are merely $n$ pairs of GR and DE scores available, creating a significant disparity. 
    This imbalance makes the estimation of $\tilde{\bbeta}$ impractical without introducing additional structural constraints. To address this challenge, we turn to an inherent characteristic in network data—community structure. Community structure reflects the tendency of nodes within the same community to exhibit similar linking patterns and usually more intense connections compared to nodes in different communities, particularly in assortative networks. By utilizing community information from the gene co-expression network, we can enhance our predictive modeling and address the over-parameterization issue effectively.


    Suppose there are $K$ communities among the genes. We use $\bZ\in\{0,1\}^{n\times K}$ to denote the community membership such that $Z_{i, k} = 1$ if gene $i$ is in the $k$th community. Given the motivation above, we impose a block structure of $\tilde{\bbeta}$ according to the community structure
    \begin{align}
    \label{equ: SBM_beta}
        \tilde{\bbeta} = \bZ \bbeta \bZ^\top,
    \end{align}
    for $\bbeta \in \RR^{K\times K}$. The diagonal entries of $\bbeta$ reflect the within-community causal effects, while the off-diagonal entries represent the between-cluster effect strengths. We remark that different from the conventional stochastic block model for network data, the core coefficient matrix $\bbeta$ is not necessarily symmetric as the effect $\beta_{k_2, k_1}$ from the GR scores in the $k_1$th community to the DE scores in the $k_2$th community are directional, for $k_1, k_2\in [K]$. In genetic studies, genes within functional modules or pathways often behave similarly and exhibit coordinated expression and regulation \citep{langfelder2008wgcna}. For instance, in an etiologically active community, genes frequently co-express, leading to a strong positive causal effect from GR to DE scores, whereas this relationship may be weaker or even opposite in an etiologically inactive community. Therefore, imposing block structure in the regression coefficient matrix, or module-level genetic studies in general, could increase statistical power and biological insight compared to individual gene analyses, especially for single-cell data analysis with limited sample sizes \citep{adamson2016multiplexed,jin2020vivo}. With the block structure of $\tilde{\bbeta}$ in \eqref{equ: SBM_beta}, the neighborhood regression model \eqref{equ: neighborhood_regression} can be rewritten as 
    \begin{align}
        \by = (\bZ\bbeta \bZ^\top *\bA)\bx + \bepsilon, \label{eq:model}
    \end{align}
    where $\epsilon = (\epsilon_1, \ldots, \epsilon_n)^\top$ is the noise vector. The proposed network-based neighborhood regression framework is illustrated in \Cref{fig:CLSE}. 
    Compared to fused lasso in high-dimensional regression \citep{hesamian2024explainable} and graph-fused lasso \citep{yu2025graph}, model \eqref{eq:model} directly accommodates neighborhood effects and community structure by sparsity and community-wise constancy, which reduces model complexity and enables computational feasibility. Relaxation to soft sparsity and constancy would be interesting and requires further exploration.
    
    In what follows, we assume the community membership $\bZ$ is known.
    If it is unknown, it can be exactly recovered with high probability, which is called strong consistency \citep{yunpeng2012Consistency}, or recovered up to a vanishing fraction with high probability, which is called weak consistency \citep{yunpeng2012Consistency}, from the network data under relatively mild conditions, provided the averaged degree of the nodes diverges as $n$ goes to infinity.
    We will inspect the issue of membership misspecification in \Cref{thm: mis_Z}.

    \begin{figure}[!t]
    \centering\includegraphics[width=\textwidth]{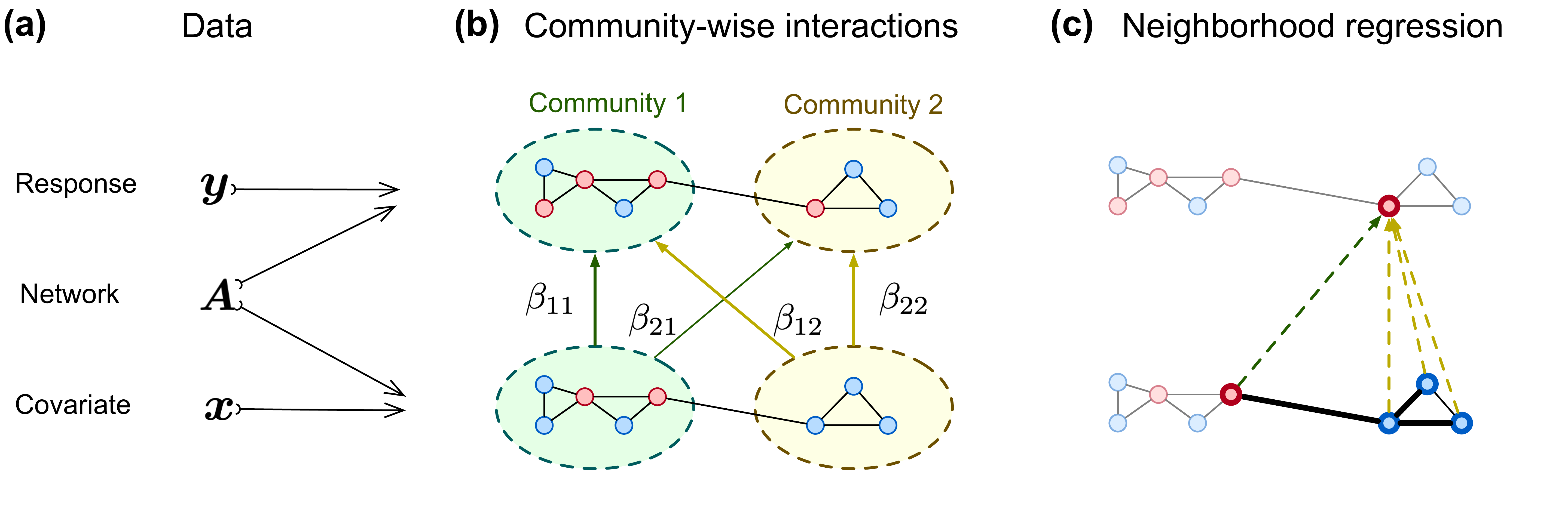}
    \caption{Network-based neighborhood regression.
    (a) The covariate $\bx$ and response $\by$ are observed for each node in a network $\bA$.
    (b) The community-wise interactions are modeled by the coefficient matrix $\bbeta$.
    (c) The conditional mean of a particular response is the average of the covariates in its neighborhood weighted by the community-wise interaction strengths. 
    } \label{fig:CLSE}
\end{figure}
    
\subsection{Community-wise least square estimation}
   To estimate the core community-level coefficient $\bbeta$ in \eqref{equ: SBM_beta}, we consider the following least square objective as in conventional linear regression framework
\begin{equation}
\label{equ: objective}
    \cL(\bbeta) = \frac{1}{2n}\|\by - (\bm{Z\bbeta Z}^\top *\bA)\bx \|_2^2.
\end{equation}
At first glance, the objective function \eqref{equ: objective} looks a bit counterintuitive since the parameters to be estimated $\bbeta$ and the predictor vector $\bx$ respectively take the places of the design matrix and regression coefficients in the classical multiple linear regression setup. 
Fortunately, simple algebra yields that the objective function \eqref{equ: objective} can be decomposed into the weighted average of $K$ non-overlapping objective functions in that
\begin{align*}
    \mathcal{L}(\bbeta) = \sum_{k = 1}^K \frac{n_k}{n}\mathcal{L}_k(\bbeta_{k, .}), \text{ with }     \mathcal{L}_k(\bbeta_{k, .}) &= \frac{1}{2n_k} \| \by * \bZ_{., k} - \bM_k \bbeta_{k, .} \|_2^2, 
\end{align*}    
where $n_k = \sum_{i=1}^n Z_{i, k}$ is the size of the $k$th community and the design matrix for the linear regression objective with transformed response $\by * \bZ_{., k} $ is
\begin{align}
    \bM_k &= \diag\left(\bZ_{\cdot, k}\right)\left(\one_n \bx^\top * \bA\right) \bZ .\label{eq:Mk}
\end{align}

We say the objective functions $\mathcal{L}_k$'s are non-overlapping because $\mathcal{L}_k$ is solely a function of the $k$th row of the core coefficient matrix $\bbeta$ but not any other entries. Therefore, minimizing $\mathcal{L}(\bbeta)$ is equivalent to minimizing each $\mathcal{L}_k(\bbeta_{k, .})$ individually, for $k\in [K]$. Intuitively, only those $y_i$'s with $i$ inside the $k$th community contain useful information in estimating $\bbeta_{k, .}$ while every $x_j$ can contribute to $y_i$. This is why, in $\mathcal{L}_k(\cdot)$, we can zero out the responses not in the $k$th community by $\by*\bZ_{., k}$ and the corresponding rows in the design matrix $\bM_k$ by a factor $\diag(\bZ_{., k})$, while the information of every $x_j$ is encoded in the $j$th column of the factor $\one_n \bx^\top * \bA$ in $\bM_k$. Moreover, $(\one_n \bx^\top * \bA)\bZ$ aggregates the samples' neighborhood effects according to their community memberships. One can essentially eliminate $n-n_k$ zero entries of $\by * \bZ_{., k} - \bM_{k} \bbeta_{k, .}$ to save memory and speed up the computation, but we stick with this expression to avoid introducing excessive notations.
Before ending this paragraph, we remark that the community-wise least square estimation can be readily extended to multiple regression settings through tensor decomposition, as demonstrated in \Cref{app:multiple}.

\subsection{Fixed-design analysis}
From the discussion above, it is sufficient to minimize $\mathcal{L}_k(\bbeta_{k, .})$ for each community individually. 
We, therefore, coin the proposed optimization problem as community-wise least square optimization. The corresponding solution, \emph{{C}ommunity-wise {L}east {S}quare {E}stimator} (CLSE), is defined as 
\begin{align}
    \hat{\bbeta} = (\hat{\bbeta}_{1, .}, \ldots, \hat{\bbeta}_{K, .})^\top \text{ with } \hat{\bbeta}_{k, .} = \argmin_{\bbeta_{k, .}\in\RR^{K}} \cL_k(\bbeta_{k, .}). \label{eq:hbeta}
\end{align}
In the fixed-design scenario, both the covariate $\bx$ and the network $\bA$, and hence $\bM_k$, are treated as deterministic while only the responses are random, we can recover various desired properties as in the classical multiple linear regression setups. For instance, the objective function $\mathcal{L}_k$ is convex with respect to $\bbeta_{k,\cdot}\in \RR^{K}$ and is strongly convex with parameter $\lambda_{\min}(\bM_k^\top \bM_k)$ if $\bM_k$ has full column rank. 
Moreover, a closed-form expression for the CLSE estimator \eqref{eq:hbeta} can be derived as in the following proposition.

\begin{proposition}[Stationary point]\label{prop:stationary-point}
    The stationary point of $\partial \cL_k(\bbeta_{k, .})/ \partial\bbeta_{k, .}$ is the solution to the following system of normal equations:
    \begin{align}
    \label{equ:normal_equation}
    \bM_k^{\top}\bM_k \bbeta_{k, .} = \bM_k^{\top}\by,
    \end{align}
    for any $k\in[K]$. Particularly, when $\bM_k$ has rank $K$, it follows that $\hat{\bbeta}_{k,\cdot} = (\bM_k^\top\bM_k)^{-1}\bM_k^{\top}\by$. 
\end{proposition}

Proposition \ref{prop:stationary-point} resembles the solution to the least square estimator in multiple linear regression.
When $\bM_k$ is singular, there are infinitely many solutions to \eqref{equ:normal_equation}.
In this case, one may use the minimum $l_2$-norm solution $\hat{\bbeta}_{k, .} = (\bM_k^\top \bM_k)^{\dag} \bM_k^\top \by$ \citep{hastie2022surprises,patil2024optimal}.
Comparing to imposing general low-rank, sparse, and/or fusion structure on the regression coefficient matrix as in \citet{zhang2020mixed, zhang2023generalized}, where flexibility for data-adaptive learning increases while careful optimization algorithm design is required, the proposed block structure on the regression coefficient matrix directly enables an explicit solution for the CLSE.

Denote the Hessian matrix of $\cL_k(\cdot)$ by $\bH_k = \bM_k^\top \bM_k$ and the underlying true parameter by $\bbeta_{k, 0} =\{ \EE(\bH_k\mid \bM_k)\}^{-1}\EE(\bM_k^{\top}\by\mid \bM_k) = \bH_k^{-1} \bM_k^{\top} \EE(\by\mid \bM_k)$. 
A consequence of \Cref{prop:stationary-point} is the asymptotic normality of the estimator, as stated in the following corollary.  
\begin{corollary}[Asymptotic normality]\label{cor:asym-normal}
    Conditioned on $\bA$ and $\bx$, for any $k\in[K]$, assume that $\rank(\bM_k)=K$ and $\epsilon_i$'s are independent and identically distributed with $\EE(\epsilon_i) =0$ and $\Var(\epsilon_i) = \sigma_k^2$, for those $i$'s in the $k$th community. 
    It then follows that
    \begin{align*}
        \hbeta_{k,\cdot} \to \cN_K( \bbeta_{k,0}, \sigma_k^2\bH_k^{-1}),
    \end{align*}
    in distribution, where $\cN_K$ stands for $K$-dimensional multivariate normal distribution. 
\end{corollary}

\Cref{cor:asym-normal} can be readily obtained from classical results for linear regression estimator \citep{eicker1963asymptotic, fahrmeir1985consistency}. Under the full column rank assumption on $\bM_k$, both \Cref{prop:stationary-point} and \Cref{cor:asym-normal} suggest that the CLSE $\hbeta_{k, .}$ shares the same form and properties as the least squares estimator for classical multiple linear regression.

The fixed-design analysis mandates the design matrix $\bM_k$ to be non-degenerate. In our formulation, the rank of the Hessian is influenced not only by the covariates $\bx$ but also by the network $\bA$, which is typically a single, noisy sample. For instance, gene co-expression networks are derived by binarizing correlation matrices with measurement errors \citep{liu2015network}. When $\bA$ is considered as a random network, its symmetry causes the rows of $\bM_k$ to be neither independent nor identically distributed. This inherent sampling randomness and unique structure of $\bA$ complicate the rank conditions, consistency, and optimality analysis, setting our neighborhood regression framework apart from conventional random-design linear regression. Therefore, an in-depth investigation of random-design analysis for our proposed framework is essential, as elaborated in the next section.

\section{Random-design analysis}\label{sec:random-design}

    \subsection{Assumptions}
    We begin by introducing several technical assumptions. First, similar to the majority of literature for network data analysis \citep{Abbe2018recent, lee2019review}, we assume the network data $\bA$ follows the stochastic block model up to deterministic self-loops. 

    \begin{assumption}[Stochastic block model]\label{asm:sbm} 
        Assume that $A_{i, j} = A_{j, i}$'s are independent Bernoulli random variables with success probability $P_{i, j}$'s, for $i< j $, where $\bP = \bZ \bB \bZ^\top$ and $\bB \in \mathbb{R}^{K\times K}$ is a symmetric community level probability matrix.
    \end{assumption}

    In \Cref{asm:sbm}, the probability matrix $\bB$ determines the connectivity strengths between and within communities, and $P_{i, j}$ only depends on the community memberships of vertexes $i$ and $j$, for $i\ne j$. Independent connectivities may not faithfully capture the actual relationships among genes, and Bahadur representation has been proposed by \citet{yuan2021community} to model the dependent connectivities for network data under the stochastic block model. In addition, The network sparsity can be characterized by $s_n = \max_{k_1, k_2 \in [K]} B_{k_1, k_2}$. As real-life networks are usually sparse, community detection and other tasks in network data can be feasible only when the average degree of the vertexes diverges. We, therefore, require the following assumption on the network sparsity. 
    \begin{assumption}[Network sparsity]\label{asm:spasity}
       Assume that the network sparsity satisfy $s_n \ge\frac{\log n }{n}$. 
    \end{assumption}
    Similar definitions of network sparsity and assumption have been popularly employed in network data analysis, including hypergraph networks \citep{zhen2023community} and multi-layer networks \citep{lei2020consistent,xu2023covariate}. 

   Unlike most literature on network data analysis that requires the community sizes to be balanced or asymptotically balanced (i.e., $n_k$'s are asymptotically of the same order), we only require the following much weaker assumption that tailors the proposed network-based neighborhood regression framework. 

  \begin{assumption}[Community sizes]\label{asm:community-sizes}
    There exists some positive constant $\delta$ such that 
    \begin{align}
    \label{equ: assum_3}
    \|\bx\|_{\infty} \max_{k\in[K]}\left({\frac{\log n}{n_ks_n(\|\bx\|^2 + \|\bx^{(k)}\|^2)}} \right)^{1/2}\le \delta,
    \end{align}
    where $\bx^{(k)}= \bZ_{., k} * \bx$ for $k \in [K]$.
  \end{assumption}
   \Cref{asm:community-sizes} is relatively mild, and we can understand it through the following examples. 
   As the first example, if $|x_i|$'s are upper bounded and lower bounded away from $0$, then the quantity on the left hand side of \eqref{equ: assum_3} is essentially of the order $\{{\log (n)}/{(n_k n s_n)}\}^{1/2} = \cO(n_k^{-{1/2}})$ according to \Cref{asm:spasity}. 
   As another example, if $x_1, \ldots, x_n$ are independent standard normal random variables, and thus $\|\bx\|^2$ is a Chi-square random variable with degree $n$, then it can be shown that $\|\bx\|_{\infty} \le c \log n $ and $\vert\|\bx\|^2 - n\vert \le c ({n \log n})^{1/2}$ with high probability for some constant $c$. 
   Therefore, the left hand side of \eqref{equ: assum_3} is $\cO_p(\log n \{\log(n)/ (n_k n s_n)\}^{1/2}) = \cO_p(n_k^{-1/2}\log n )$. In both examples, \eqref{equ: assum_3} holds as long as the community size $n_k$ diverges faster than $(\log n)^2$. Moreover, both examples allow the upper bound $\delta$ to vanish. We also remark that including the $\|\bx^{(k)}\|^2$ in the denominator can better accommodate imbalanced communities because the optimal $k$ in the left-hand side of (\ref{equ: assum_3}) does not necessarily lead to $n_{\min}$ in the denominator, the smallest community size.

   The next assumption concerns the tail behavior of the additive noise $\bepsilon$. 

      \begin{assumption}[Additive noise]\label{asm:epsilon}
        Conditional on $\bx$ and $\bA$, we assume $\epsilon_1,\ldots,\epsilon_n$ are 
        independent centered sub-exponential random variables with uniform parameters $(\sigma_{\epsilon}^2, b_{\epsilon})$. Precisely, 
        \begin{equation*}
        \EE(e^{t \epsilon_i}) \le e^{t^2 \sigma_{\epsilon}^2 /2}, \text{  for any } \vert t \vert < b_{\epsilon}^{-1} \text{ and } i\in [n]. 
        \end{equation*}
    \end{assumption}

    The conditional independent assumption on $\epsilon_i$'s still allows for the dependence between $\epsilon_i$ and $\epsilon_j$, for any $i\ne j$. This makes the assumption more realistic to genomic data, where the genes are adjacency with each other in a chromatin that has a 3-dimensional spatial structure. Alternatively, we may directly model the dependence of $\epsilon_1, \ldots, \epsilon_n$ as a weakly dependence sequence satisfying $\tau$-mixing property \citep{merlevede2011bernstein}. The sub-exponential assumption on $\epsilon_i$'s is mild and includes the classes of bounded random variables and sub-Gaussian random variables \citep[Definition 2.7]{wainwright2019high}. Also, \Cref{asm:epsilon} does not require $\epsilon_i$'s to be identically distributed, while most classic regression setups do. Finally, to better illustrate the rate of convergence, sometimes it will be convenient to assume the following tali bound for the maximum deviation of $\bx$'s entries. 
        \begin{assumption}
        \label{assum: x}
        Conditional on $\bA$, $x_i$'s are independent zero-mean random variables, and there exists a constant $\gamma$ and a quantity $\kappa_n$ vanishing with $n$ such that with probability at least $1-\kappa_n$, it holds that $\|\bx\|_{\infty} \le \gamma (\log n)^{1/2}$. 
  \end{assumption}
  Note that the $x_i$'s are not necessarily independent under \Cref{assum: x}, and a wide range of classes of distributions, such as sub-exponential random variables,  shall satisfy such exponentially decaying probabilistic tail bound. 

    \subsection{Non-asymptotic analysis of the Hessian}
    \label{sec: non-asym}
    In this section, we study the spectral property of the Hessian $\bH_k$ by investigating its mean $\EE(\bH_k)$ and its concentration behavior of $\bH_k$ to $\EE(\bH_k)$. We first decompose $\bH_k$ as
    \begin{align}
    \label{equ:I_1+I_2+I_3+I_4}
        & \bH_k = \bI_1+ \bI_2 +  \bI_3 +\bI_4,   
    \end{align}
    where $\bI_1 =  \bZ^\top \left[ \bx \bx^\top * (\bA-\EE(\bA)) \diag(\bZ_{., k}) (\bA-\EE(\bA))\right]\bZ$ is a matrix quadratic form of $\bA - \EE(\bA)$,  $\bI_2 = \bZ^\top \left[ \bx \bx^\top * (\bA-\EE(\bA)) \diag(\bZ_{., k}) \EE(\bA)\right]\bZ$ has zero-mean, $\bI_3 = \bI_2^\top$,
    and $\bI_4 =\bZ^\top \left[ \bx \bx^\top * \EE(\bA)\diag(\bZ_{., k}) \EE(\bA) \right]\bZ $ is deterministic and positive semi-definite when $\bx$ is given. Moreover, we can further decompose $\bI_1 = \bS_1 + \bS_2$, where 
    \begin{align}
    \label{equ:S1+S2}
    &\bS_1  = \sum_{k_1=1}^K \sum_{k_2 =1}^K  \sum_{\psi_i = k_1} \sum_{\psi_{i^\prime} = k_2} \sum_{\psi_j = k} x_i x_{i^\prime} \ind\left(i\ne i^\prime \right)(\bA - \EE(\bA) )_{i, j} (\bA - \EE(\bA))_{i^\prime, j} \be_{k_1} \be_{k_2}^\top, \text {and }\nonumber\\
    &\bS_2 = \sum_{k_1=1}^K  \sum_{\psi_i = k_1} \sum_{\psi_j = k} x_i^2  (A_{i, j} - \EE(A_{i, j}))^2 \be_{k_1} \be_{k_1}^\top. 
    \end{align}
   Herein, $\be_{j}\in\RR^n$ is a unit vector whose $j$th entry being one, $\psi_{i} = \argmax_{k \in[K]}Z_{i,k}$ is the community assignment of the $i$th sample. Note that the notation $\sum_{\psi_i = k_!}$ is a shorthand notation for $\sum_{i:, \psi_i = k_1}$, which sums terms across all the indexes $i$ in the $k_1$-th community, similarly for $\sum_{\psi_{i^\prime} = k_2}$ and $\sum_{\psi_j = k}$. Clearly, $\bS_1$ is a zero-mean symmetric matrix, and $\bS_2$ is a diagonal matrix that has non-negative diagonals. Intuitively, the spectral information of $\bH_k$ is mainly encoded in $\bS_2$ and $\bI_4$, while the zero-mean terms shall have small spectral norms. In fact, it follows from the above decomposition that $\EE(\bH_k) = \EE(\bS_2) + \bI_4$, and 
    \begin{align}
        \bH_k - \EE(\bH_k)  = \bS_1 + \{\bS_2 - \EE(\bS_2)\} + (\bI_2 + \bI_2^{\top})  \label{eq:decomp}
    \end{align}
    Since $\bI_4$ is positive semi-definite, we can verify straightforwardly that 
    \begin{align}
    \lambda_{\min}\left(\EE(\bH_k)\right)\ge \lambda_{\min}\left(\EE(\bS_2)\right) \ge \min_{k^\prime \in [K]}   B_{k^{\prime}, k}(1 - B_{k^{\prime}, k})(n_k - 1)\|\bx^{(k^{\prime})}\|^2. \label{eq:lam-min-EHk}
    \end{align}

    From \eqref{eq:lam-min-EHk}, the lower bound of $\lambda_{\min}(\EE(\bH_k))$ shall be of order $\Omega( n_k n_{\min}s_n)$ with $n_{\min} = \min_{k^\prime \in [K]} n_{k^\prime}$ if $B_{k^\prime, k}$ is of order $\Omega(s_n)$.
    This indicates that $\EE(\bH_k)$ is guaranteed to have full rank.
    If the perturbation of $\bH_k - \EE(\bH_k)$ can be further controlled, we are able to infer that $\bH_k$ is nonsingular with high probability. 
    The following theorem provides a careful perturbation analysis of $\bH_k - \EE(\bH_k)$. 
    
  \begin{theorem}[Fisher information concentration]
  \label{thm:eig-Hk}
  Under \Crefrange{asm:sbm}{asm:community-sizes}, for the Hessian $\bH_k=\bM_k^{\top}\bM_k$ with $\bM_k$ defined in \eqref{eq:Mk}, it holds that $\|\bH_k - \EE(\bH_k)\| \leq r_k$ with probability at least $1 - 2K(C_1n_k +2C_1 + 2)/n^2$, for some universal constant $C_1$, where 
  \begin{align*}
        r_k &=  \alpha(\delta)(C_1 + {s_n}^{1/2}) ({ \|\bx\|^2 + \|\bx^{(k)}\|^2 })^{1/2}\left\{\left({\sum_{k^\prime=1}^K \|\bx^{(k^\prime)}\|_\infty^2}\right)^{1/2}+\frac{\|\bx\|}{{n}^{1/2}}\right\} s_n ({n_k n})^{1/2} \log n \\
        & + 2^{3/2} \left(\frac{2}{3^{1/2}}C_1 {s_n}^{1/2} + 1\right)  \max_{k^\prime\in [K]}\|\bx^{(k^\prime)}\|_4^2({n_ks_n\log n})^{1/2}+  4(C_1+1) \|\bx\|_{\infty}^2 \log n,
    \end{align*}
  and $\alpha(\delta)= \{ (8\delta + 4({4\delta^2 + 9})^{1/2}) \}/3$. 
  \end{theorem}
  The universal constant $C_1$ comes from the decoupling constant in \citet{de1995decoupling}, and our result and technical proof do not induce any other unclear constant. In addition, $\delta$ comes from \Cref{asm:community-sizes} and $\alpha(\delta)$ will decrease to $4$ if $\delta$ vanishes. As non-asymptotic analysis is conducted, the probabilistic upper bound looks a bit complicated. The next corollary details the asymptotic order of $\|\bH
_k -\EE(\bH_k)\|$. 

\begin{corollary}\label{cor:subGau-x}
    Under Assumptions \ref{asm:sbm}-\ref{asm:community-sizes} and \ref{assum: x}, there exists a universal constant $C_2$, such that 
    \begin{align*}
    \|\bH_k - \EE(\bH_k)\| \le C_2 \gamma^2 K^{1/2}s_n n n_k^{1/2}  (\log n)^2, 
    \end{align*}
     with probability at least $1 - 2K(C_1n_k +2C_1 + 2)/n^2 - \kappa_n$. 
\end{corollary}
Since $\lambda_{\min}(\bH_k) =\Omega(s_n n_k n_{\min})$, \Cref{cor:subGau-x} allows us to conclude concentration if $n (\log n)^2 = O(n_k^{1/2} n_{\min})$ and $K = O(1)$. It then follows from Weyl's inequality that $\lambda_{\min}(\bH_k)$ is asymptotically of the same order as $\lambda_{\min}\left(\EE(\bH_k)\right)$.  

Based on the decomposition \eqref{eq:decomp}, the proof for \Cref{thm:eig-Hk} relies on a decoupling approach \citep{de1995decoupling} to bound the matrix quadratic form  $\bS_1$, the usual Bernstein's inequality together with the union bound to bound the diagonal matrix $\bS_2 - \EE(\bS_2)$ and matrix Bernsten's inequality \citep{tropp2012user} to bound  $\bI_2$, which are done in \Crefrange{lem:Nk-tI1}{lem:Nk-tS2}, \Cref{lem:Nk-S2}, and \Cref{lem:Nk-I2} in the supplementary materials, respectively.
  Our proof technique is related to, but substantially different from, the technique in \citet{lei2023bias} or Hanson-Wright type inequality \citep{Hanson2013} for matrix quadratic forms. This is because every summand in the decomposition \eqref{equ:I_1+I_2+I_3+I_4} contains a left factor $\bZ^{\top}$ and a right factor $\bZ$ that aggregates the random variables according to their community memberships, while the Theorems in \citet{lei2023bias} work for the matrix quadratic form $\bF \bG\bF^\top$ for a deterministic matrix $\bG$ and a random matrix $\bF$ that has independent entries or is symmetric with independent upper triangle entries. Apparently, the appearance of $\bZ$ makes $\bZ^\top[\bx \bx^\top * (\bA - \EE(\bA))]\diag(\bZ_{\cdot,k})$, the random matrix in $\bI_1$ for example, neither symmetric nor have independent entries. Additionally, the Hadamard factor $\bx \bx^\top$ also adds an extra layer of difficulty to derive the probabilistic concentration bound. All of these require a subtle and careful analysis.  

      \begin{remark}[Matrix quadratic form]
        In the simple scenario that $\bx = \bm{1}_n$, if we upper bound $\|\bI_1\|$ by $n_{\max} \|(\bA-\EE(\bA)) \diag(\bZ_{., k}) (\bA-\EE(\bA))\|$ and employ Theorem 5 of \citet{lei2023bias} to upper bound $\|(\bA-\EE(\bA)) \diag(\bZ_{., k}) (\bA-\EE(\bA))\| = \Op(ns_n \log n)$, it leads to $\|\bI_1\| = \Op(n_{\max} n s_n \log n )$ with $n_{\max} = \max_{k^\prime \in [K]} n_{k^\prime}$, which fails to conclude concentration since  $\lambda_{\min}(\EE(\bH_k))= \Omega(n n_{\min}s_n)$. Therefore, it is vital to make full use of the aggregation structure in $\bI_1$ while studying its concentration behavior. 
    \end{remark}

    \subsection{Consistency}
    \label{sec: consistency}
    Under the random-design setting, the oracle coefficient $\bbeta_{k, .}^*$ is defined as the solution to the population-level normal equation.
    More specifically, by taking expectation on both sides of \Cref{equ:normal_equation}, we obtain
  \begin{align*}
   \bbeta^*_{k, .} = \EE(\bH_k)^{-1}\EE(\bM_k^{\top}\by),
   \end{align*}
   where the expectation is taken with respect to both $\bA$ and $\bepsilon$. Conditional on $\bx$, $\bA$ and the event that $\lambda_{\min}(\bH_k)$ diverges, when $\Var(\epsilon_i) = \sigma^2$, for all $i$ in the $k$-th community, Corollary 2 provides asymptotic normality and thus concludes the asymptotic consistency of $\hat{\bbeta}_{., k}$ without concerning the rate of convergence. For general random-design analysis, the next theorem shows that the estimator $\hbeta_{k, .}$ is an unbiased of $\bbeta^*_{k, .}$, and provides the rate of convergence.
    
    \begin{theorem}[Unbiasedness and Consistency]\label{thm:consistency}
    Under Assumption \ref{asm:epsilon}, it holds that $\EE(\hbeta_{k, .}) = \bbeta^*_{k, .}$. Moreover, with probability at least $1 - (2K)/n^2$, we have
    \begin{align*}
      \|\hat{\bbeta}_{k, .} - \bbeta_{k, .}^* \|\le  2 \left( 2K /\lambda_{\min}(\bH_k)\right)^{1/2}(\sigma_{\epsilon}+ b_{\epsilon})\log n.
    \end{align*}
     Additionally under the assumptions in \Cref{thm:eig-Hk}, we have
    \begin{align*}
        \|\hat{\bbeta}_{k, .} - \bbeta_{k, .}^* \| & \le 2 \left( 2K\right)^{1/2}(\sigma_{\epsilon}+ b_{\epsilon})\left\{ \min_{k^\prime \in [K]}   B_{k^{\prime}, k}(1 - B_{k^{\prime}, k})(n_k - 1)\|\bx^{(k^{\prime})}\|^2  -  r_k \right\}^{-1/2} \log n.
    \end{align*}
    with probability at least  $1 - 2K(C_1n_k +2C_1 + 3)/n^2$. 
\end{theorem}

Clearly, the upper bound of $\|\hat{\bbeta}_{k, .} - \bbeta_{k, .}^* \|$ comes from the lower bound of $\lambda_{\min}(\bH_k)$ and the sub-exponential concentration behaviors on the $\epsilon_i$'s. The following corollary elaborates on this non-asymptotic upper bound in terms of asymptotic order. 
\begin{corollary}
\label{cor: root_n}
    Under Assumptions \ref{asm:sbm}-\ref{assum: x}, if additionally $B_{k^\prime, k} = \Omega(s_n)$, $\min_{k^\prime \in [K]} \|\bx^{(k^\prime)}\| = \Omega(n_{\min})$, and $n (\log n)^2 = O(n_k^{1/2} n_{\min})$, then there exists a constant $C_3$, such that 
    \begin{align*}
    \|\hat{\bbeta}_{k, .} - \bbeta_{k, .}^* \| \le C_3 \frac{ K^{1/2}(\sigma_{\epsilon} + b_{\epsilon})\log n }{(s_n n_k n_{\min})^{1/2}},
    \end{align*}
     with probability at least $1 - 2K(C_1n_k +2C_1 + 3)/n^2 - \kappa_n$. 
\end{corollary}
If $s_n\asymp 1$ and the community sizes are balanced such that $n = O(n_{\min})$, surprisingly, \Cref{cor: root_n} suggests linear consistency, $\|\hat{\bbeta}_{k, .}-{\bbeta}_{k, .}^*\|=\cO(n^{-1})$, instead of the canonical root-$n$ consistency for regression.
This shows the blessing of incorporating network neighborhood information. Intuitively, after incorporating the network data, the effective sample size increases from $n$ to $n + s_n n(n-1)/2$, which is of the order $s_n n^2$, while the number of parameters to be estimated is of constant order. This makes linear consistency possible as $(s_n n^2)^{1/2} = s_n^{1/2}n$ is linear in $n$ if $s_n$ is of constant order. This also suggests that the network sparsity $s_n$ plays an important role as the variation of $s_n$ smoothly transforms $\hat{\bbeta}_{k, .}$ from canonical regime to blessing of neighborhood information regime.

We next investigate the influence of community detection error on the convergence rate of $\hat{\bbeta}_{k, .}$ when the community membership is unknown. 
Let $\alpha_n$ be the number of mis-clustering nodes from community detection, which is potentially diverging.
In most consistent community detection literature, $\alpha_n = \Op\left(\log n/s_n\right)$; see, for example, \citet{zhen2023community}.
We have the following result.

\begin{theorem}[Community membership misspecification]
\label{thm: mis_Z}
Denote $\hat{\bbeta}_{k, .}^{\mis}$ be the community-wise least square estimator with $\alpha_n$ nodes misspecified.  Under the condition of \Cref{cor: root_n}, there exists some positive constant $C_4$ and $C_5$, such that with probability at least $1 - C_4K/n - \kappa_n$, we have
\begin{align*}
     \|\hat{\bbeta}_{k, .}^{\mis} - \bbeta_{k, .}^* \| \le C_5 \frac{ K^{1/2}(\sigma_{\epsilon} + b_{\epsilon} + \gamma s_n \sqrt{n \alpha_n} \log n)\log n }{(s_n n_k n_{\min})^{1/2}}. 
\end{align*}
\end{theorem}
Clearly, misspecification of the community membership leads to an extra term $\gamma s_n \sqrt{n \alpha_n} \log n$ in the numerator, depending on the network sparsity. When $\alpha_n = \cO(1)$, if the network is extremely sparse with $s_n = \cO (\frac{\log n}{n})$ or relatively dense with $s_n = \cO(1)$, \Cref{thm: mis_Z} maintains the canonical $\sqrt{n}$-consistency, up to logarithm factor. For moderate sparse network, for example $s_n = \cO(\frac{\log n}{\sqrt{n}})$, \Cref{thm: mis_Z} improves the $\sqrt{n}$-consistency to $n^{3/4}$-consistency, up to logarithm factor. The intuition behind this is as follows. When the community membership of $1$ node is misspecified, it leads to potential misspecification of $n$ regression coefficients in $\bZ^\top \bbeta \bZ * \bA$. However, a sparser network helps to reduce the actual number of misspecified regression coefficients, though it also carries weaker signals.
Before ending this subsection, in analogy to the Gauss-Markov Theorem for linear regression, we have the following theorem. 

    \begin{theorem}[Community-wise best linear unbiased estimator]\label{thm:BLUE} 
    Under \Crefrange{asm:sbm}{asm:epsilon}, further assume that the variance of $\epsilon_i$'s are homogeneous within the $k$th community, then $\hat{\bbeta}_{k, .}$ is the best linear unbiased estimator of $\bbeta_{k, .}^*$, i.e., for any linear unbiased estimator $\bar{\bbeta}_{k}$, we have $\Var( \hat{\bbeta}_{k,\cdot} ) \preceq \Var( \bar{\bbeta}_{k} )$. 
    \end{theorem}
    Herein, a linear estimator means the estimator is linear in the response $\by$.

\subsection{Minimax optimality}
To investigate the minimax optimality of the community-wise least square estimator $\hat{\bbeta}_{k, .}$, we first introduce a class of data distributions:
\begin{align*}
& \mathcal{P}_k(P_X, P_{\bA}, \sigma) = \left\{ P_{X, \bA, Y}: \bx \sim P_X, \bA \sim P_{\bA}, \bbeta^*_{k, .} \in \RR^K, \EE(\bepsilon) = \zero_n, \EE(\epsilon_i^2) \le \sigma_{\max}^2\text{ for } i\in [n]\right\}, 
\end{align*}
for $k \in [K]$, where $\sigma = (\sigma_1, \ldots, \sigma_n)^\top$ and $\sigma_{\max} = \max_{i \in [n]} \sigma_i$, for $\sigma_i = \EE(\epsilon_i^2)$, $i\in [n]$. Define the following discrepancy between any estimator $\bar{\bbeta}_{k, .}$ and $\bbeta^*_{k, .}$ 
\begin{align*}
\mathcal{E}(\bar{\bbeta}_{k, .}) = \|\bar{\bbeta}_{k, .} -\bbeta_{k, .}^* \|_{\bSigma}^2,  
\end{align*}
for any positive definite matrix $\bSigma$, where $\|\cdot\|_{\bSigma}$ denotes the Mahalanobis distance of a vector to the origin. When $\bSigma = \bI_n$, $\mathcal{E}(\bar{\bbeta}_{k, .})$ characterizes the estimation error of $\bar{\bbeta}_{k, .}$. When $\bSigma =\EE(\bH_k)$, $\mathcal{E}(\bar{\bbeta}_{k, .})$ captures the generalization performance of $\bar{\bbeta}_{k, .}$ in that $\mathcal{E}(\bar{\bbeta}_{k, .}) = \cR_k(\bar{\bbeta}_{k, .}) - \cR_k(\bbeta_{k, .}^*) $ represents the excess risk, for the risk function $\cR_k(\bbeta_{k, .}) = \EE(\|\bY * \bZ_{., k} - \bM_{k, .}\bbeta_{k, .}\|_2^2)$. The minimax expected discrepancy for the $k$th sub-problem is then defined as 
\begin{align*}
\inf_{\bar{\bbeta}_{k, .}} \sup_{\bbeta^*_{k, .} \in \mathcal{P}_k(P_X, P_{\bA}, \bsigma^2)}\EE\left\{\mathcal{E}(\bar{\bbeta}_{k, .})\right\}, 
\end{align*}
where the infimum is taken with respect to all estimators from the data. The next theorem provides an exact expression for the minimax expected discrepancy. 
\begin{theorem}[Minimax optimality]\label{thm:minimax}
The exact minimax risk can be expressed as 
\begin{align*}
\inf_{\hat{\bbeta}_{k, .}\in\RR^K} \sup_{\bbeta^*_{k, .} \in \mathcal{P}_k(P_X, P_{\bA}, \bsigma^2)}\EE\left(\mathcal{E}(\hat{\bbeta}_{k, .})\right) = \sigma_{\max}^2\tr\left(\bSigma \EE (\bH_k^{-1} ) \right) . 
\end{align*}
\end{theorem}
Since $\tr(\bSigma \bS^{-1})$ is strictly convex with respect to $\bS$ in the positive definite cone, Jensen's inequality yields that $\tr(\bSigma \EE(\bH_k^{-1}))\ge \tr(\bSigma (\EE (\bH_k))^{-1})$. Taking $\bSigma = \bI$, we have
\begin{align*}
\inf_{\hat{\bbeta}_{k, .}} \sup_{\bbeta^*_{k, .} \in \mathcal{P}_k(P_X, P_{\bA}, \bsigma^2)}\EE\big( \|\hat{\bbeta}_{k, .} -\bbeta_{k, .}^* \|^2\big) \ge \sum_{k^\prime=1}^K\frac{\sigma_{\max}^2}{\lambda_{k^\prime}(\EE(\bH_k))} \ge \frac{\sigma_{\max}^2}{\lambda_{\min}(\EE(\bH_k))}, 
\end{align*}
which matches up with the probabilistic upper bound in \Cref{thm:consistency}. This indicates the non-asymptotic analysis in the previous sections is sharp.

\section{Simulation}\label{sec:simu}

    \subsection{Impact of network neighborhood information}\label{subsec:simu-net}

    The first simulation study analyzes how the network information helps in the estimation. Specifically, for any $n\in\{100,200,\ldots,1000\}$ and $K\in\{2,3,4\}$, we begin with randomly and uniformly sampling the community memberships for $n$ genes, resulting in the membership matrix $\bZ \in \{0, 1\}^{n\times K}$. 
    Next, the network $\bA$ is generated according to the stochastic block model stated in \Cref{asm:sbm} with $B_{k_1, k_2}\sim\mathrm{Uniform}(0,0.5)+0.5*\ind\{k_1=k_2\}$ and $B_{k_2, k_1}=B_{k_1, k_2}$, for $1\leq k_1<k_2\leq K$.
    We subsequently generate the response $\by$ according to model \eqref{eq:model}, where the entries of both the covariate $\bx$ and the coefficient matrix $\bbeta\in\RR^{K\times K}$ are sampled independently from the standard normal distribution, and the additive noises are drawn from $\epsilon_1,\ldots,\epsilon_n\overset{\iid}{\sim}\cN(0,0.5^2)$.
    We then evaluate 4 methods that utilize the network information in different ways: (i) CLSE that utilizes the network neighborhood information appropriately, (ii) CLSE with $\bA$ replaced by $\bI_n$ which completely ignores the interactions between vertices, (iii) CLSE with $\bA$ replaced by $\one_n\one_n^{\top}$ which includes all potential interactions among vertices, and (iv) netcoh, the network cohesion estimator proposed by \citet{li2019prediction}, which minimizes the penalized loss function $
    \mathcal{L}(\balpha, \beta) = \|\by - \balpha - \beta \bx\|^2 + \lambda \balpha^\top \left [\diag(\bA \bm{1}_n) - \bA\right] \balpha$.
    
    Clearly, it equips each node with an individual intercept to adjust a linear regression model and penalizes these intercepts according to the graph Laplacian. In practice, we select the hyperparameter by cross-validation over a grid of $100$ regularization parameters from $10^{-3}$ to $10$ uniformly spaced in the $\log_{10}$-scale.

    In all our experiments of \Cref{sec:simu,sec:case-study}, we employ a variation of SCORE method \citep{Jin2015fast, ke2023score} to estimate the community memberships for each node when it is not directly available. 
    Precisely, we compute $\bU\in\RR^{n\times K}$ that contains the eigenvectors of $\bA$ corresponding to the $K$ leading singular values, subsequently normalize each row of $\bU$ to have unit $l_2$-norm, and finally apply K-means algorithm to obtain the cluster assignments $\hat{\bZ}$. 
    The number of communities is determined by identifying the elbow point in the singular value distribution of $\bA$ as in \citep{CoauthorshipJin2016, rohe2016co}. 
    A scree plot of the leading singular values of the real-life autism gene co-expression network is provided in \Cref{fig:data} (b).
    After obtaining any estimator $\hat{\bbeta}$ by community-wise least square estimation, the predicted values for $\by$ is given by $\hat{\by} = (\hat{\bZ}\hat{\bbeta} \hat{\bZ}^\top *\bA)\bx$.
    We inspect the estimation error and the prediction error, defined as
   \[\mathrm{Err}_{\mathrm{est}}= \frac{1}{K^2}\|\hat{\bQ}^\top\hat{\bbeta}\hat{\bQ}-\bbeta^*\|_{\fro}^2 \text{ and } \mathrm{Err}_{\mathrm{pred}} = \frac{1}{n}\|\hat{\by}-\by\|_2^2, \text{ where } \hat{\bQ} =  \argmin_{\bQ \in \mathbb{G}_K} \|\hat{\bZ}\bQ - \bZ\|_{\fro}^2.\]
    Herein, $\mathbb{G}_K$ denotes the set of $K\times K$ permutation matrices. Clearly, $\hat{\by}$ is invariant to the permutation among the communities, while $\hat{\bbeta}$ is not. That is why we need to search for the best permutation that minimizes the Hamming distance of the community memberships encoded in $\hat{\bZ}$ and $\bZ$, and then define the estimation error accordingly.
    
    Results in the logarithm scale are shown in \Cref{fig:simu-1}, where the trends for small values have been zoomed in while those for large values have been zoomed out.
    As we can see, both the estimation and prediction errors of CLSE for different community sizes are decreasing in the sample size $n$ and approaching zero quickly.    
    On the other hand, without properly utilizing the network information $\bA$, both the estimation and prediction errors cannot be controlled even with large sample sizes.
    In summary, this simulation study showcases the asymptotic consistency of the CLSE estimator and suggests the necessity of utilizing the network neighborhood information appropriately.

    \begin{figure}
        \centering
        \includegraphics[width=0.9\textwidth]{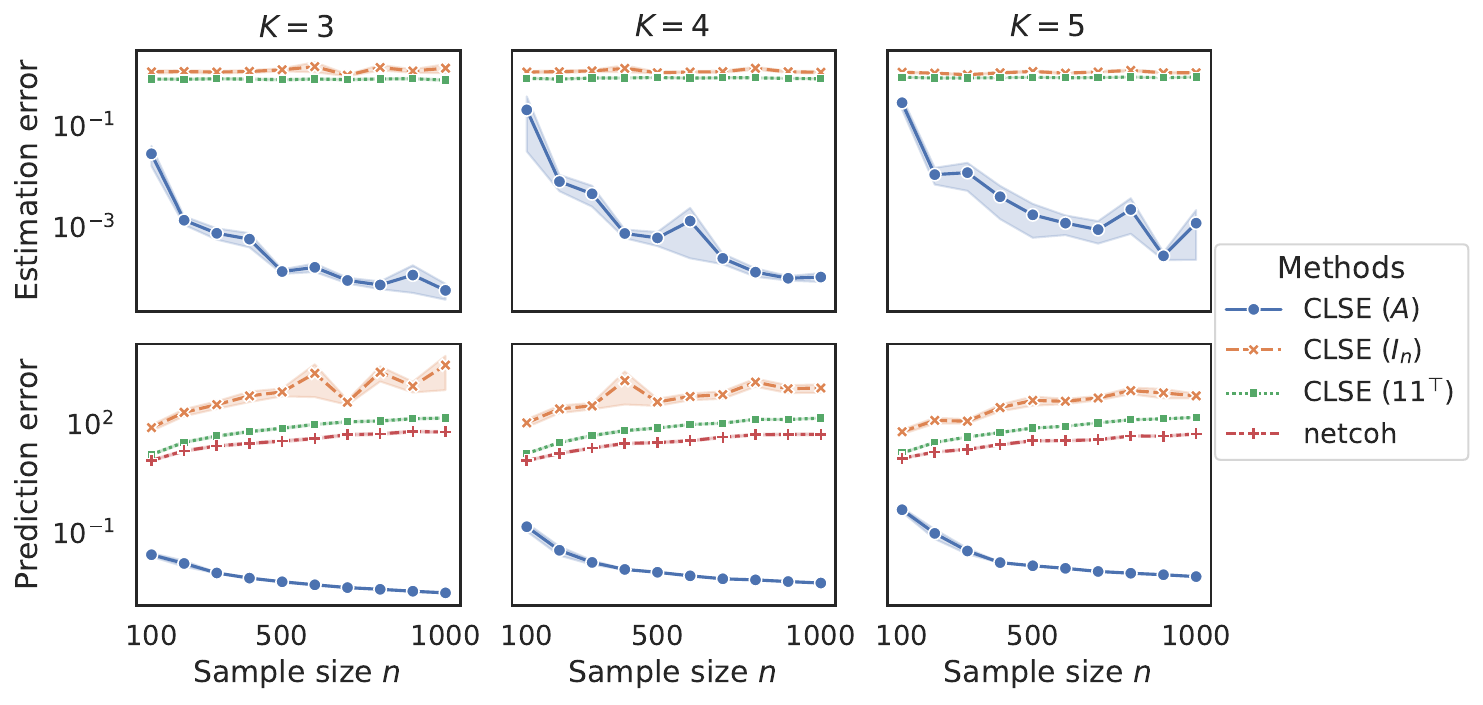}
        \caption{
        Estimation and prediction errors in the log scale of experiments in \Cref{subsec:simu-net} with varying numbers of communities.        
        The shaded regions represent the standard errors around the average values computed over 200 simulated datasets.
        For netcoh, only the prediction errors are available and presented.}
        \label{fig:simu-1}
    \end{figure}

    \subsection{Impact of community structure}\label{subsec:simu-coef}

    The second simulation study considers various regression coefficient models,  including (i) full model: $\bbeta\in\RR^{K\times K}$, (ii) row model: $\bbeta= \one_K\bbeta_0^{\top}$  for $\bbeta_0\in\RR^K$, and (iii) singleton model: $\bbeta=\beta_0 \one_K\one_K^{\top}$ for $\beta_0\in\RR$. Denote by $\hbeta^{\full} $ the solution to \eqref{eq:hbeta} corresponding to the full model. The counterparts to the row and singleton models are derived as follows. 
    Under the setting of row model, \eqref{eq:hbeta} reduces to a multiple linear regression problem:
        \begin{align*}
            \hbeta_0^{\row} = \argmin_{\bbeta_0\in\RR^K} \frac{1}{2n}\|\by - (\bZ\one_K\bbeta_0^{\top}\bZ^{\top} * \bA)\bx\|_2^2 
            = (\bZ^{\top} \diag(\bx) \bA^2 \diag(\bx) \bZ)^{-1} \bZ^{\top} \diag(\bx) \bA \by,
        \end{align*}
        which yields the row estimator $\hbeta^{\row} = \one_K({\hbeta}_0^{\row})^{\top}$.
        Under the singleton model setup, \eqref{eq:hbeta} reduces to a simple linear regression problem:
         \begin{align*}
            \hat{\beta}_0^{\sing}= \argmin_{\beta_0\in\RR} \frac{1}{2n}\|\by - (\bA\bx) \beta_0\|_2^2 = \frac{\bx^\top \bA \by}{\bx^\top \bA^2 \bx} .
        \end{align*}
        which yields the singleton estimator $\hbeta^{\sing} = \hat{\beta}_0^{\sing}\one_K\one_K^{\top}$.
        Different estimators utilize the community information at different levels.

    \begin{figure}
        \centering
        \includegraphics[width=0.9\textwidth]{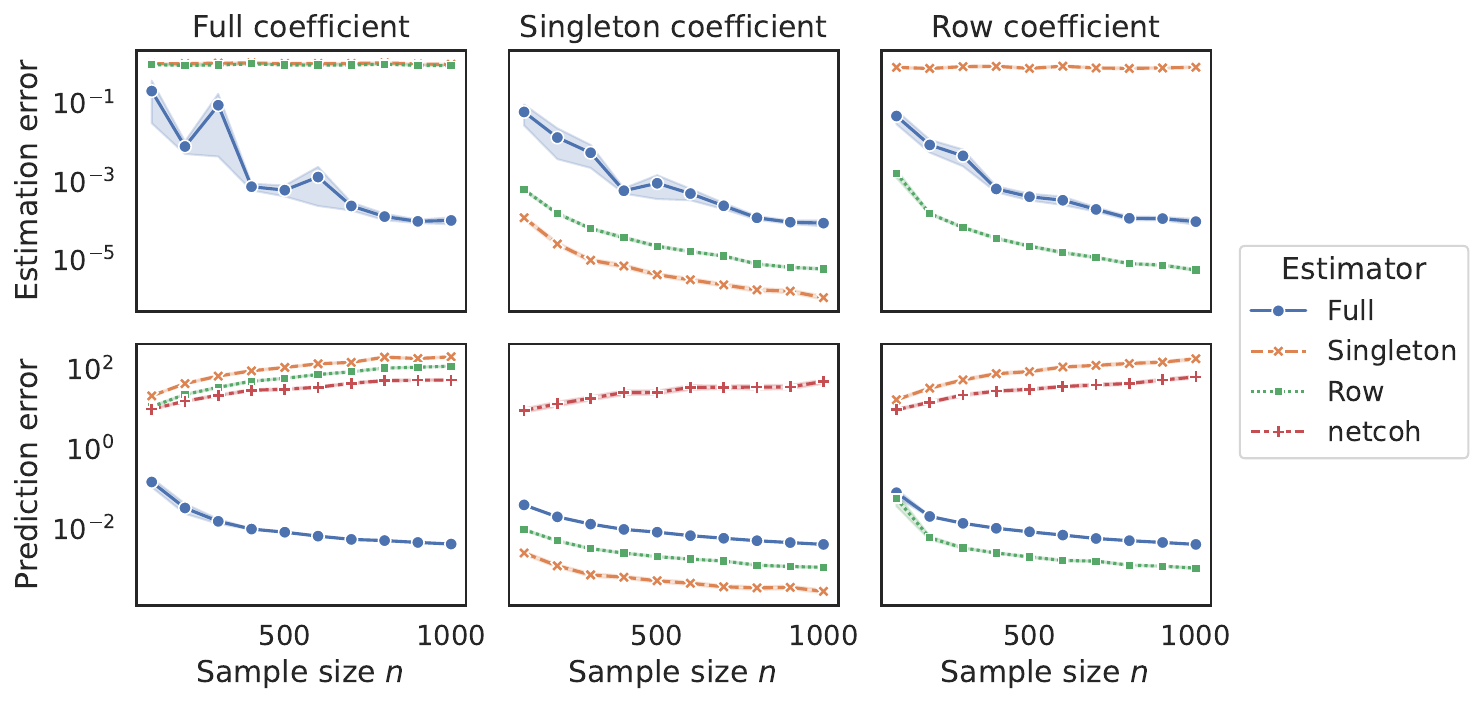}
        \caption{
        Estimation and prediction errors in the log scale of experiments in \Cref{subsec:simu-coef} with different coefficient structures.        
        The shaded regions represent the standard errors around the average values computed over 200 simulated datasets.}
        \label{fig:simu-2}
    \end{figure}

    We compare the estimation and prediction errors of the three estimators in all three generating schemes of regression coefficients with $K =4$ and varying $n \in \{100, \ldots, 1000\}$.
    The results are shown in \Cref{fig:simu-2}.
    It is expected that each estimator works best under its own well-specified setting.
    However, both the singleton and row estimators are sensitive to model mis-specifications, while the full estimator adapts well due to its generality, with errors tending to zero relatively fast. These results suggest that the full estimator is capable of utilizing the community information, and it does not suffer much when only a partial of this information is relevant. 
    Note that netcoh fails to provide satisfying results in all scnarios, including in the singleton model setup. Though it penalizes the intercept vector such that nodes sharing similar linking pattern have similar intercept adjustments, it ignores the more subtle neighborhood information, and  
 the predicted value of $y_i$ given by netcoh is constructed from the estimate of $\alpha_i, \beta$ and $x_i$ only, without considering the covariates of the nodes in the neighborhood of node $i$.

\section{Autism spectrum disorder genetic association}\label{sec:case-study}
    \subsection{Background}
    Autism spectrum disorder (ASD) primarily stems from genetic variations, either inherited or arising spontaneously in individuals. 
    This genetic diversity plays a crucial role in ASD's prevalence, with de novo exonic variations being particularly valuable for linking specific genes to the disorder \citep{fu2022rare}.
    The Transmission and De Novo Association (TADA) method \citep{he2013integrated} has been pivotal in pinpointing genes susceptible to ASD by analyzing mutation frequencies in family trios, leading to the identification of numerous ASD-associated genes, yet many remain undiscovered. 
    Despite the identification of thousands of genes with differential expression in ASD \citep{gandal2022broad}, there is minimal overlap between DE genes and the GR  genes deemed significant by TADA in the two studies.
    In this section, we attempt to utilize the gene co-expression networks to integrate these disparate data sources and disentangle the impact of the GR scores from TADA on the DE scores.

    \begin{figure}[!t]
        \centering
        \includegraphics[width=\textwidth]{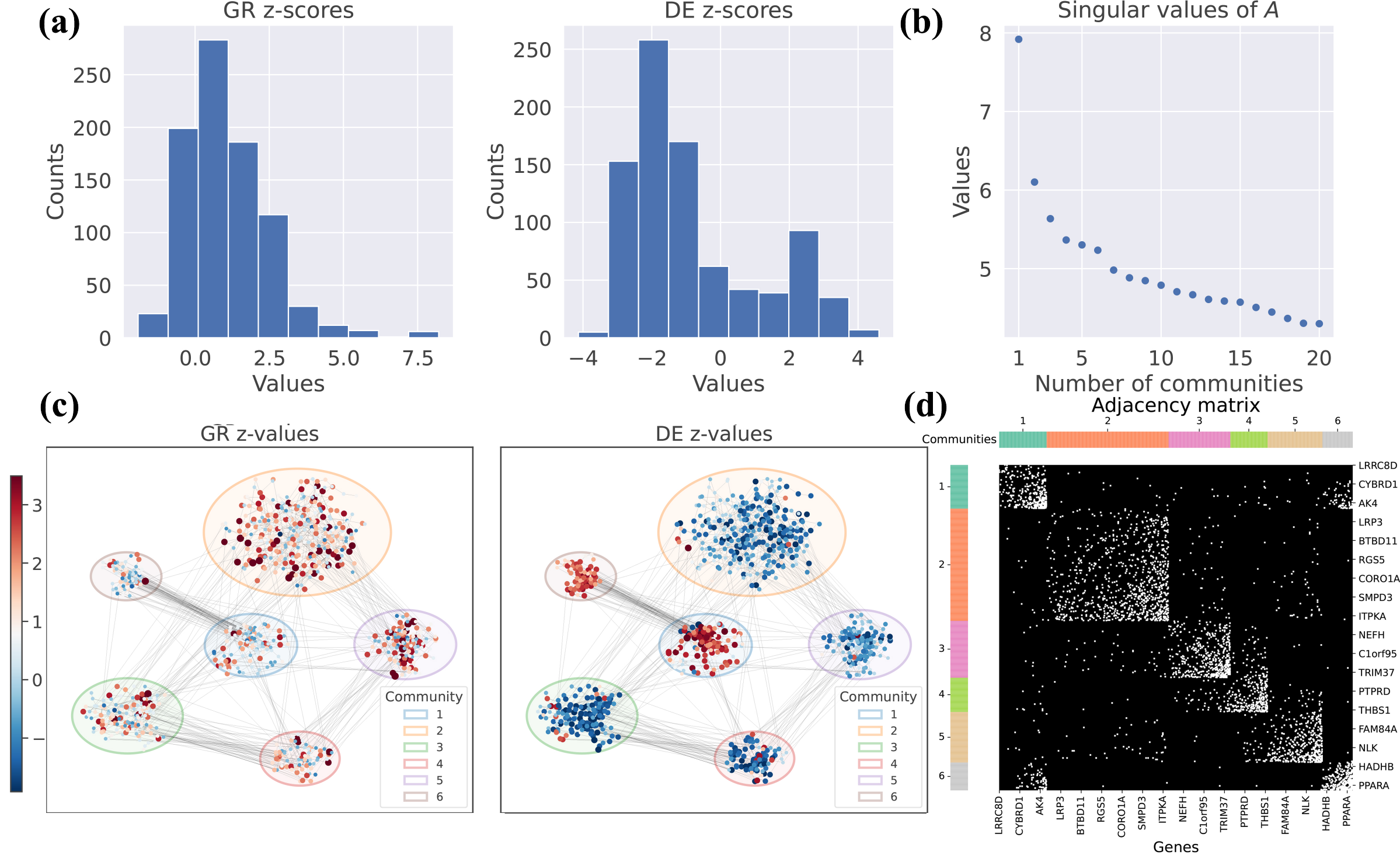}
        \caption{Visualization of ASD data and detected communities.
        (a) The histograms of GR one-sided z-scores ($\bx$) and DE two-sided z-scores ($\by$) for 864 substantial autism genes.
        (b) The scree plot of the singular values of $\bA$. 
        (c) The gene co-expression network colored by z-values and grouped by estimated communities.
        (d) The adjacency matrix $\bA$ colored by connectivity (white for 1 and black for 0) and ordered by estimated assignments.
        }
        \label{fig:data}
    \end{figure}
    
    \subsection{Data and preprocessing}

    We use two types of genomics data:
    (1) The DE and GR test statistics are originally obtained from the differential expression analysis by \citet{gandal2022broad} and the TADA analysis by \citet{fu2022rare}, respectively.
    We use the GR and DE z-values as the covariate $\bx\in\RR^n$ and the response $\by\in\RR^n$, respectively.
    The histograms of the two scores are shown in \Cref{fig:data}(a).
    (2) A whole cortex gene expression data (bulk RNA-sequencing data) on neurotypical individuals is also available from the previous study by \citet{gandal2022broad}.
    Based on gene expression data, \citet{liu2015network} use the partial neighborhood selection algorithm to obtain a sparse network $\bA\in\{0,1\}^{n\times n}$ of approximately scale-free form; though one can also use other networks, such as the protein-protein interaction networks.
    Finally, we restrict the analysis to a subset of 864 substantial autism genes, which is identified by using the partial neighborhood selection (PNS) algorithm of DAWN \citep{liu2014dawn}.

    Based on the network $\bA$, we first perform community detection to uncover the genes' community memberships. We first visualize the 20 leading singular values of $\bA$ in \Cref{fig:data}(b).
    It is clear that the singular values decay quickly and become smaller than 5 after the first $6$ leading singular values, suggesting that there shall be $6$ communities. Consequently, we select $K=6$. We have then utilized the same community detection method as in \Cref{sec:simu} to obtain the estimated cluster membership $\hat{\bZ}$. Grouping the genes into the detected communities, we observe a clear block structure of the adjacency matrix, as shown in \Cref{fig:data}(d). 

    The estimated cluster membership matrix allows us to visually compare the two sources of statistical evidence.
    As shown in \Cref{fig:data}(c), for most of the communities, the GR and DE scores share similar patterns.
    The genes in a cluster with enriched GR scores in terms of absolute values, such as communities 1 and 2, typically also have large DE scores, while the genes in communities 5 and 6 have both small GR and DE scores simultaneously.
    These suggest that the genes in such clusters may be positively regulated by their genetic variations of the same gene module.
    On the other hand, genes in communities 3 and 4 have moderate GR scores while much larger DE ones.
    Given the scarce evidence of genetic variation, the genes in these two communities are likely to be regularized by other genetic modules.
    These observations motivate us to analyze the interplay of different gene modules and understand how genetic variations affect gene expressions among different clusters by performing network-based neighborhood regression coupled with community-wise analysis.

    \begin{table}[!t]\scriptsize
        \centering\setlength{\tabcolsep}{3.5pt}
        \scalebox{0.9}{
        \begin{tabularx}{0.98\textwidth}{crrrrrr}
            \toprule
            \multirow{2}{2cm}{\textbf{Target Comm.} ($\by$)} & \multicolumn{6}{c}{\textbf{Source Comm.} ($\bx$)} \\\cmidrule(lr){2-7}
            & 1 & 2 &3 & 4 & 5 & 6\\\midrule
        \multirow{2}{*}{1} & \cellcolor{magenta!10} $0.325 \pm 0.070$ & $-0.037 \pm 0.347$ & $-1.057 \pm 1.026$ & $-0.171 \pm 1.440$ & $-0.193 \pm 0.572$ & \cellcolor{magenta!10} $0.379 \pm 0.146$\\
         & \cellcolor{magenta!10} (${*}{*}{*}$) 0.0000 &  0.9155 &  0.3029 &  0.9057 &  0.7362 & \cellcolor{magenta!10}    (${*}{*}$) 0.0092\\ \addlinespace[0.5ex] 
          \multirow{2}{*}{2} & $-0.256 \pm 0.270$ & \cellcolor{Cyan!10} $-0.197 \pm 0.014$ & \cellcolor{Cyan!10} $-0.537 \pm 0.148$ & $0.041 \pm 1.096$ & \cellcolor{Cyan!10} $-0.222 \pm 0.050$ & $-2.067 \pm 6.931$\\
         &  0.3422 & \cellcolor{Cyan!10} (${*}{*}{*}$) 0.0000 & \cellcolor{Cyan!10} (${*}{*}{*}$) 0.0003 &  0.9704 & \cellcolor{Cyan!10} (${*}{*}{*}$) 0.0000 &  0.7655\\ \addlinespace[0.5ex] 
          \multirow{2}{*}{3} & $-0.947 \pm 0.650$ & $-0.179 \pm 0.183$ & \cellcolor{Cyan!10} $-0.305 \pm 0.027$ & $-0.107 \pm 0.107$ & $0.390 \pm 16.968$ & $1.511 \pm 2.472$\\
         &  0.1448 &  0.3262 & \cellcolor{Cyan!10} (${*}{*}{*}$) 0.0000 &  0.3135 &  0.9816 &  0.5411\\ \addlinespace[0.5ex] 
          \multirow{2}{*}{4} & $0.359 \pm 0.428$ & $-0.724 \pm 0.568$ & \cellcolor{Cyan!10} $-0.224 \pm 0.098$ & \cellcolor{Cyan!10} $-0.346 \pm 0.043$ & \cellcolor{Cyan!10} $-0.274 \pm 0.098$ & \cellcolor{Cyan!10} $-1.925 \pm 0.362$\\
         &  0.4021 &  0.2023 & \cellcolor{Cyan!10}       (${*}$) 0.0221 & \cellcolor{Cyan!10} (${*}{*}{*}$) 0.0000 & \cellcolor{Cyan!10}    (${*}{*}$) 0.0051 & \cellcolor{Cyan!10} (${*}{*}{*}$) 0.0000\\ \addlinespace[0.5ex] 
          \multirow{2}{*}{5} & $0.054 \pm 0.494$ & \cellcolor{Cyan!10} $-0.227 \pm 0.046$ & $-0.500 \pm 2.795$ & \cellcolor{Cyan!10} $-0.239 \pm 0.093$ & \cellcolor{Cyan!10} $-0.154 \pm 0.014$ & $-0.520 \pm 0.917$\\
         &  0.9123 & \cellcolor{Cyan!10} (${*}{*}{*}$) 0.0000 &  0.8581 & \cellcolor{Cyan!10}       (${*}$) 0.0101 & \cellcolor{Cyan!10} (${*}{*}{*}$) 0.0000 &  0.5707\\ \addlinespace[0.5ex] 
          \multirow{2}{*}{6} & $0.407 \pm 0.265$ & $0.356 \pm 0.239$ & \cellcolor{magenta!10} $2.164 \pm 0.134$ & $1.158 \pm 4.500$ & $-1.277 \pm 0.867$ & \cellcolor{magenta!10} $0.490 \pm 0.068$\\
         &  0.1239 &  0.1361 & \cellcolor{magenta!10} (${*}{*}{*}$) 0.0000 &  0.7969 &  0.1406 & \cellcolor{magenta!10} (${*}{*}{*}$) 0.0000\\ \addlinespace[0.5ex] 
             \bottomrule
        \end{tabularx}
        }
        \caption{Estimation and inference results of  $\hbeta$. The $k$th row of the table corresponds to $\hbeta_{k,\cdot}$.
        Within each cell, the point estimate and the estimated standard deviation are given on top of the cell, while the significance level and the corresponding p-value are given at the bottom.
        For the significance levels, (${*}{*}{*}$), (${*}{*}$), and (${*}$), indicate that the p-value locates in $[0,0.001]$, $(0.001,0.01]$, and  $(0.01,0.05]$, respectively.
        The significant positive and negative coefficients are highlighted in magenta and cyan, respectively.
        }
        \label{tab:test}
    \end{table}

    \subsection{Neighborhood regression on ASD genetic association}

    Based on the evidence $(\bx,\by)$ and network $\bA$, along with the estimated cluster membership $\hat{\bZ}$, we compute the estimated values of the CLSE and that given by netcoh. We find that the in-sample prediction error of CLSE is $0.5781$, which is substantially smaller than that of netcoh (3.0719) in terms of their respective objective functions, showing that CLSE delivers more trustworthy result than netcoh. This also reflects that the genes actually co-express with each other, which can be efficiently characterized by the network-based neighborhood regression model, while netcoh fails to leverage the GR scores in the neighborhood to the DE score of any particular gene.
    
    We next focus on CLSE only and perform individual hypothesis testing on $H_0:\beta_{k_1,k_2}=0$ versus $H_1:\beta_{k_1,k_2}\neq 0$ for $k_1,k_2\in[K]$. 
    The heteroskedasticity-consistent (HC) standard errors \citep{mackinnon1985some} of the estimators are used to compute p-values. 
    All of the results are summarized in \Cref{tab:test}.
    From \Cref{tab:test}, we observe that the GR scores have significant effects on the DE scores within the same community, which is expected because if a gene module is associated with genetic variations of ASD, then the gene expression levels of this module will also be affected.
    Though the intra-cluster interaction of the two modalities is important, we are more interested in the inter-cluster interaction, as it will shed new light on how one gene module regulates the others. 
    By exploring the regulatory mechanisms and potential influences between distinct groups of genes, we can gain insights into the broader network dynamics at play.
    For this purpose, we further restrict our analysis to the two most significant inter-cluster coefficients $\beta_{4,6}$ and $\beta_{6,3}$, corresponding to directional effects from community $6$ to community 4 and from community $3$ to community $6$, respectively. 
    
    By matching the gene modules identified in \citet[Fig. 7]{gandal2022broad}, clusters 3 and 4 contain genes that are mostly enriched in three cell types: excitatory neurons, inhibitory neurons, and oligodendrocytes, and cluster 6 contains genes that are mostly enriched in astrocytes, which have a potential impact on neuronal function and connectivity and are critical in the pathology of ASD \citep{gandal2022broad}.

    \begin{figure}[!t]
        \centering
        \includegraphics[width= 0.9\textwidth]{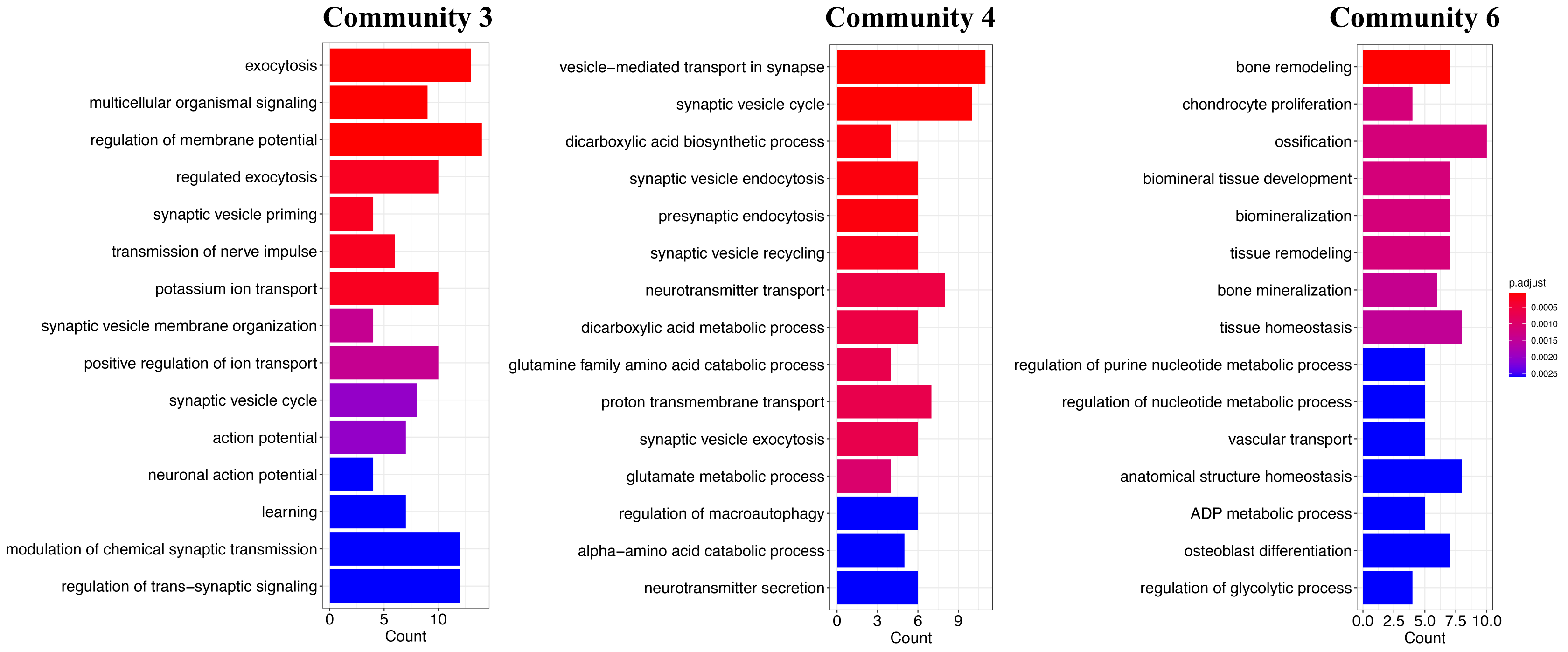}
        \caption{Top gene ontology terms for genes in community 3, 4, and 6.}
        \label{fig:GO}
    \end{figure}
    
    From the gene ontology (GO) analysis result in \Cref{fig:GO}, community 6 is enriched for GO terms related to bone remodeling, ossification, biomineralization, and regulation of nucleotide metabolism, while community 4 is enriched for GO terms involved in synaptic vesicle cycling, neurotransmitter transport, and proton transport processes at the synapse.
    The estimated coefficient in \Cref{tab:test} suggests a negative impact of community 6 on community 4, which could potentially be explained by the fact that excessive bone remodeling and mineralization processes (community 6) may disrupt normal synaptic functions (community 4) by altering the ionic balance or metabolic processes required for neurotransmission for ASD.

    In addition, community 3, which is enriched for GO terms related to exocytosis, synaptic vesicle priming, nerve impulse transmission, and ion transport regulation, has a positive impact on community 6. This positive impact could be due to the fact that regulated exocytosis and ion transport processes (community 3) may facilitate the release of factors or signaling molecules that promote bone remodeling, ossification, and biomineralization (community 6).
    The interpretation is sorely based on the GO term enrichment of biological processes.
    The specific molecular mechanisms underlying the observed impacts would require further experimental validation and investigation.

\subsection{ A novel adjusted coefficient of determination}
Since our primary goal is to understand the module-wise interactions, we employ $\beta_{k_1, k_2}$ to capture the average causal effect from GR scores among genes in one community to the DE scores of those in another. Biologically, genes within functional modules or pathways often behave similarly and exhibit coordinated expression and regulation. Admittedly, not all genes within a module behave identically, while the assumption reasonably leverages genetic heterogeneity and allows for meaningful and broad biological insights.

To verify the validity of the assumption in a data-driven way, analogous to linear regression, we proposed the following adjusted coefficient of determination $R^{2, net}_{adj}$ for network-based neighborhood regression
\begin{align*}
R^{2, net}_{adj} = 1 -  \frac{n - 1}{n - K^2} \cdot \frac{\sum_{k=1}^K \|\by^{(k)} - \hat{\by}^{(k)}\|^2}{\|\by - \hat{\by}^{(ols)})\|^2},
\end{align*}
where $\by^{(k)}, \hat{\by}^{(k)}$ and $\hat{\by}^{(k, ols)}$ are the actual response, the predicted response by the proposed method, and predicted response by ordinary least square in the $k$-th community, respectively. Note that since the proposed network-based neighborhood regression does not contain an intercept, we chose the ordinary least square as the baseline for fair comparison. The value of $R^{2, net}_{adj}$ for the real data being studied is greater than $0.5867$. This implies that when the block structure of the coefficient matrix is imposed, over $58.67$\% of the uncertainty of the actual response $\by$ can be explained by the neighborhood regression model, demonstrating the rationale behind the block-wise effects.

\section{Discussion}

This paper incorporates network-based neighborhood information to bridge the predictor and response in the regression setup. Potential extensions of the current framework include allowing multivariate predictors and multivariate responses, extending the neighborhood regression to generalized neighborhood regression for binary or counting responses as in generalized linear models \citep{du2023simultaneous}, and considering more general network structures, such as weighted, directed, multi-layer \citep{lei2023bias}, and hypergraph networks \citep{zhen2023community}. In the supplement, we suggest a potential extension for multiple covariates, while its statistical property remains unclear. Moreover, it will be interesting to slightly relax the exact stochastic block model for the network data and the block structure for the coefficient matrix to allow more heterogeneity and flexibility among the entities, such as using the latent space model for network modeling. Besides, the proposed neighborhood regression framework is also closely related to the sum aggregator of graph neural network (GNN) \citep{xu2018powerful}, which is provably the most expressive among a number of classes of GNNs and is as powerful as the Weisfeiler-Lehman graph isomorphism test.
Exploring connections between the proposed method and other aggregation operators in GNN with heterogeneous structures presents a promising avenue for future research. 


\section*{Supplementary material}
The supplementary material includes all technical proofs, necessary lemmas, and Python implementations of the paper's numerical experiments.

\section*{Acknowledgement}
We thank the associate editor, three anonymous referees, and the reproducibility reviewer, whose constructive comments and suggestions have led to significant improvements of the article. 
We would also like to express our great gratitude to Kathryn Roeder, Bernie Devlin, Jinjin Tian, and Maya Shen 
for enlightening discussion and for sharing the processed data.
This paper was initiated when the authors met at the 2023 {\it IMS Young Mathematical Scientists Forum — Statistics and Data Science}, hosted by the National University of Singapore (NUS). We truly appreciate NUS for the warm invitation.

The authors report there is no competing interest to declare.

{
\spacingset{1.1}\small
\putbib[references]}
\end{bibunit}

\clearpage
\appendix
\begin{bibunit}[apalike]


\begin{center}
\bf\Large{Supplementary material for\\
\titletext}
\end{center}

\bigskip

This serves as an appendix to the paper ``\titleRLB.''
The organization for the appendix is as below:
\begin{itemize}[leftmargin=5mm]

    \item In \Cref{app:intercept}, we extend the neighborhood regression model to incorporate the intercept term.


    \item In \Cref{app:multiple}, we extend the neighborhood regression model to incorporate multiple covariates.
    
    \item In \Cref{sec:proof}, we present the technical proofs of all the theoretical results.
    
    \item In \Cref{app:lemmas}, we provide supporting lemmas used in \Cref{sec:proof}.
\end{itemize}    
\bigskip

\section{Extension to include intercepts}\label{app:intercept}



Recall that in the multiple linear regression model, $y = \beta_0 + \bx^\top \bbeta + \epsilon$ with $\EE(\epsilon)=0$, we can centralize the data to reduce the model to the one without an intercept. The rationality behind this is as follows. In population level, we have
\begin{align*}
    y - \EE(y) = \beta_0 + \bx^\top \bbeta + \epsilon - (\beta_0 + \EE(\bx)^\top \beta) = \left(\bx - \EE(\bx)\right)^\top \bbeta + \epsilon.
\end{align*}
Therefore, regressing $y- \EE(y)$ on $\bx - \EE(\bx)$ results in a zero intercept. In sample level, let $(\bx_i, y_i)_{i=1}^n$ be the sample points, the sample mean $(\bar{\bx} = n^{-1}\sum_{i=1}^n\bx_i, \bar{y} = n^{-1}\sum_
{i=1}^n y_i)$ always satisfies the estimated regression function given by least square estimation. This is because, 
\begin{align*}
\frac{\partial\sum_{i=1}^n(y_i - \beta_0 - \bx_i^\top\bbeta)^2/(2n)}{\partial \beta_0} = -\frac{1}{n}\sum_{i=1}^n(y_i - \beta_0 - \bx_i^\top\bbeta) = 0
\end{align*}
implies the estimator satisfies $\hat{\beta}_0 = \bar{y} - \bar{\bx}^\top \hat{\bbeta}$. Hence, once the data have zero sample mean, we do not need to fit the intercept. 

In the proposed network-based neighborhood regression model, since the data $y_i$ only corresponds to the regression function given by $\bbeta_{k, .}$ if $i$ belongs to the $k$th community, we extend the model to include $K$ intercepts, one for each community-wise regression function. Let $\beta_0^{(k)}$ be the intercept accompanied with $\bbeta_{k, .}$.  For any node $i$ in the $k$th community, we extend the neighborhood regression model as 
\begin{align*}
y_i = \beta_0^{(k)} + \sum_{j \in N_i}\beta_{k, \psi_j} x_j + \epsilon_i = \beta_0^{(k)} + \sum_{j=1}^n \beta_{k, \psi_j} A_{i, j} x_j + \epsilon_i = \beta_0^{(k)} + \sum_{k^\prime = 1}^K \beta_{k, k^\prime} \sum_{\psi_j = k^\prime} A_{i, j} x_j + \epsilon_i. 
\end{align*}
The partial derivative of the least square objective function with respective to $\beta_0^{(k)}$ reads 
\begin{align*}
& \quad \frac{\partial \sum_{\psi_i = k}(y_i - \beta_0^{(k)} - \sum_{k^\prime = 1}^K \beta_{k, k^\prime} \sum_{\psi_j = k^\prime} A_{i, j} x_j)^2/(2n_k)}{\partial \beta_0^{(k)}}\\
& = -\frac{1}{n_k}\sum_{\psi_i = k} \left(y_i - \beta_0^{(k)} - \sum_{k^\prime = 1}^K \beta_{k, k^\prime} \sum_{\psi_j = k^\prime} A_{i, j} x_j\right). 
\end{align*}
Setting the above partial derivative to zero yields the estimator satisfies 
\begin{align*}
\hat{\beta}_0^{(k)} = \frac{1}{n_k} \sum_{\psi_i = k} y_i - \frac{1}{n_k}\sum_{k^\prime = 1}^K \hat{\beta}_{k, k^\prime} \sum_{\psi_i = k} \sum_{\psi_j = k^\prime} A_{i, j} x_j. 
\end{align*}
Therefore, if we centralize the data in such a way that 
\begin{align*}
\tilde{y}_i = y_i -  \frac{1}{n_k} \sum_{\psi_i = k} y_i, \text{ and } \tilde{x}_j = x_j - \frac{\sum_{\psi_j = k^\prime} (\sum_{\psi_i=k} A_{i, j}) x_j}{\sum_{\psi_j = k^\prime} (\sum_{\psi_i = k} A_{i, j})}, 
\end{align*}
for node $j$ belongs to the $k^\prime$ community, $k^\prime \in [K]$, the corresponding estimated intercept become zero. Herein, $\bar{y}^{(k)} = \sum_{\psi_i = k}y_i/n_k$ is the average for the responds in the $k$th community, and 
\begin{align*}
\mu_{k, k^\prime} =  \frac{\sum_{\psi_j = k^\prime} (\sum_{\psi_i=k} A_{i, j}) x_j}{\sum_{\psi_j = k^\prime} (\sum_{\psi_i = k} A_{i, j})} =  \sum_{\psi_j = k^\prime}\frac{ \sum_{\psi_i=k} A_{i, j}}{\sum_{\psi_j^\prime = k^\prime} \sum_{\psi_i = k} A_{i, j^\prime}}x_j
\end{align*}
is the weighted average for the covariates in the $k^\prime$ community, and $x_j$ is weighted by the proportion of connections between communities $k^\prime$ and $k$ made by node $j$. 

This is also reflected in the population level. Suppose that given $\bA$, the conditional mean of $x_j$'s are the same for those $j$'s within the same community. Thus, we have 
\begin{align*}
&\EE(\mu_{k, k^\prime} \vert \bA) =  \sum_{\psi_j = k^\prime}\frac{ \sum_{\psi_i=k} A_{i, j}}{\sum_{\psi_j^\prime = k^\prime} \sum_{\psi_i = k} A_{i, j^\prime}} \EE(x_j\vert\bA) = \EE(x_j\vert\bA), 
\end{align*}
for any node $j$ inside the $k^\prime$th community. 
In addition, 
\begin{align*}
\EE(y_i | \bA) = \beta_0^{(k)} + \sum_{k^\prime = 1}^K \beta_{k, k^\prime} \sum_{\psi_j = k^\prime} A_{i, j} \EE(x_j|\bA). 
\end{align*}
Therefore, 
\begin{align*}
y_i - \EE(y_i | \bA) &= \sum_{k^\prime = 1}^K \beta_{k, k^\prime} \sum_{\psi_j = k^\prime} A_{i, j} \left(x_j - \EE(x_j|\bA)\right) + \epsilon_i \bx \\
& =  \sum_{k^\prime = 1}^K \beta_{k, k^\prime} \sum_{\psi_j = k^\prime} A_{i, j} \left(x_j - \EE(\mu_{k, k^\prime} \vert \bA)\right) + \epsilon_i, 
\end{align*}
for node $i$ locates in the $k$th community.

\section{Extension to neighborhood regression with multiple covariates}\label{app:multiple}
In the case that the covariates for the nodes are multivariate, we use $p$ to denote the number of covariates and $\bX \in \RR^{n\times p}$ to denote the covariate data such that $X_{i, l}$ is the $l$th covariate of node $i$. Following the same spirit of the network-based neighborhood regression model in \Cref{sec:NNR}, we propose to model the response $y_i$ for node $i$ as 
\begin{align}
\label{equ: mutl_varaita_scalar_form}
y_i = \sum_{j\in N_i} \sum_{l=1}^p \tilde{\beta}_{i, j, l}X_{j, l} + \epsilon_i = \sum_{j=1}^n \sum_{l=1}^p A_{i, j}\tilde{\beta}_{i, j, l}X_{j, l} + \epsilon_i
\end{align}
for $i\in [n]$, where $\tilde{\bbeta} = (\tilde{\beta}_{i, j, l}) \in \RR^{n\times n \times p}$ is the regression coefficient tensor while $\bA$ and $y_i$ is the network data as defined before. In vector format, we can rewrite \Cref{equ: mutl_varaita_scalar_form} as 
\begin{align*}
 \by= \sum_{l=1}^p (\bA * \tilde{\bbeta}_{\cdot, \cdot, l}) \bX_{., l} + \bepsilon = \cM_1 \left(\bA \circ \bm{1}_p * \tilde{\bbeta}\right) \Vec(\bX)+ \bepsilon, 
\end{align*}
where $\tilde{\bbeta}_{\cdot, \cdot, l}$ is the $l$th frontal slide of $\tilde{\bbeta}$, $\circ$ is the outer product such that the $(i, j, l)$th entry of $\bA \circ \bm{1}_p$ is $A_{i, j} (\bm{1}_p)_l = A_{i, j}$, $\cM_1(\cdot)$ is the mode-1 matricization operator that stacks the mode-1 fibers of the input tensor as the columns of the output matrix, and $\vec(\cdot)$ is the vectorization operator.

To incorporate the community structure into the coefficient tensor, we assume 
\begin{align*}
\tilde{\bbeta} = \bbeta \times_1 \bZ \times_2 \bZ,
\end{align*}
for a core tensor $\bbeta \in \RR^{K \times K \times p}$. Herein, the bilinear product means 
\begin{align*}
\tilde{\bbeta}_{i, j, l} = \sum_{k_1, k_2} \beta_{k_1, k_2, l} \bZ_{i, k_1} \bZ_{j, k_2} = \beta_{\psi_i, \psi_j, l}.
\end{align*}
More general definitions for multi-linear products can refer to \citet{kolda2009tensor}. Therefore, the full model in terms of the $\bbeta$ is 
\begin{align*}
\by = \cM_1 \left(\bA \circ \bm{1}_p * \bbeta \times_1 \bZ \times_2 \bZ\right) \Vec(\bX)+ \bepsilon. 
\end{align*}

The objective function for least square estimation is 
\begin{align*}
\min_{\bbeta \in \RR^{K \times K \times p}}\|\by - \cM_1 \left(\bA \circ \bm{1}_p * \bbeta \times_1 \bZ \times_2 \bZ\right) \Vec(\bX)\|^2, 
\end{align*}
which can be separated into the following $K$ non-overlapping optimization problems
\begin{align*}
\min_{\bbeta_{k, ., .}\in\RR^{K\times p}}\left\|\by * \diag(\bZ_{., k}) - \diag(\bZ{., k})\left\{ \bm{1}_n \vec(\bX)^\top * \bm{1}_p^\top \otimes \bA \right\}(\bI_p \otimes \bZ) \vec(\bbeta_{k, ., .})\right\|^2,
\end{align*}
for $k\in[K]$.This is because 
\begin{align*}
\sum_{l=1}^p (\bA * \tilde{\bbeta}_{\cdot, \cdot, l}) \bX_{., l} &= \sum_{l=1}^p  \sum_{k=1}^K\diag(\bZ_{., k}) (\bm{1}_n \bX_{., l}^\top * \bA) \bZ \bbeta_{k, ., l}\\
& = \sum_{k=1}^K \diag(\bZ_{., k})\left\{[\bm{1}_n \bX_{., 1}^\top \ldots, \bm{1}_n \bX_{., p}^\top] * (\bm{1}_p^\top \otimes \bA)\right\}(\bI_p \otimes \bZ) \vec(\bbeta_{k, ., .})\\
& = \sum_{k=1}^K \diag(\bZ{., k})\left\{ \bm{1}_n \vec(\bX)^\top * \bm{1}_p^\top \otimes \bA \right\}(\bI_p \otimes \bZ) \vec(\bbeta_{k, ., .}), 
\end{align*}
which allows us to estimate the horizontal slices of the tensor $\bbeta$ individually.

\section{Proof of theoretical results}\label{sec:proof}

\subsection{Proof of \Cref{prop:stationary-point}}

\begin{proof}[Proof of \Cref{prop:stationary-point}]
    Note that
    \begin{align}
        L(\bbeta) = \|\by\|_2^2/(2n)  - \langle \by, (\bm{Z\bbeta Z}^\top * \bA)\bx\rangle/n +  \|(\bm{Z\bbeta Z}^\top *\bA)\bx \|_2^2/(2n). \label{eq:loss-decomp}
    \end{align}
    We next analyze the derivative for the last two terms on the right-hand side of the above display.
    For the second term in \eqref{eq:loss-decomp}, note that
    \begin{align*}
        \langle \by, (\bm{Z\bbeta Z}^\top * \bA)\bx\rangle &= \langle \by\bx^{\top}, \bm{Z\bbeta Z}^\top * \bA \rangle \\
        &= \langle \by\bx^{\top} * \bA, \bm{Z\bbeta Z}^\top  \rangle\\
        &= \langle \bZ^\top(\by\bx^{\top} * \bA) \bZ, \bbeta  \rangle.
    \end{align*}
    It then follows that
    \begin{align}
        \frac{\partial}{\partial\bbeta}[-\langle \by, (\bm{Z\bbeta Z}^\top * \bA)\bx\rangle ] &= -\bZ^\top(\by\bx^{\top} * \bA) \bZ = -
        \begin{bmatrix}
            \by^{\top}\bM_1\\
            \vdots\\
            \by^{\top}\bM_K   
        \end{bmatrix}. \label{eq:deri-3}
    \end{align}

    For the last term in \eqref{eq:loss-decomp}, we can rewrite it as 
    \begin{align*}
        \|(\bm{Z\bbeta Z}^\top*\bA)\bx \|_2^2
        & =\left\langle(\bZ \bbeta \bZ^{\top} * \bA)^{\top} (\bZ \bbeta \bZ^{\top} * \bA),\ \bx \bx^{\top}\right\rangle \\
        & =\left\langle [(\bZ \bbeta \bZ^{\top} * \bA)^{\top} \odot (\bZ \bbeta \bZ^{\top} * \bA)^{\top}]\one_n,\ \bx \otimes \bx\right\rangle \\
        & =\left\langle\left\{\left[(\bZ\bbeta^\top\bZ^{\top}) \odot (\bZ \bbeta^\top \bZ^{\top})\right] * (\bA \odot \bA) \right\} \one_n, \  \bx \otimes \bx\right\rangle \\
        & =\left\langle\left[(\bZ \bbeta^\top \bZ^{\top}) \odot(\bZ \bbeta^\top \bZ^{\top})\right] *(\bA \odot \bA),\ (\bx \otimes \bx) \one_n^{\top}\right\rangle \\
        & =\left\langle (\bZ \bbeta^\top \bZ^{\top}) \odot (\bZ \bbeta^\top \bZ^{\top}), (\bx \otimes \bx) \one_n^{\top} * (\bA \odot \bA) \right\rangle \\
        & =\left\langle[(\bZ\bbeta^\top) \odot (\bZ\bbeta^\top)] \bZ^{\top}, (\bx\one_n^\top *\bA) \odot (\bx\one_n^\top *\bA) \right\rangle \\
        & =\left\langle(\bZ \otimes \bZ)(\bbeta^\top \odot \bbeta^\top) \bZ^{\top}, (\bx\one_n^\top *\bA) \odot (\bx\one_n^\top *\bA)      \right\rangle \\
        & =\langle \bbeta^\top \odot \bbeta^\top, (\bZ^{\top} \otimes \bZ^{\top})\left[(\bx\one_n^\top *\bA) \odot (\bx\one_n^\top *\bA)\right] \bZ\rangle,
    \end{align*}
    where in the second equality we use the fact that $\vec(\bQ_1 \diag(\bw)\bQ_2^\top) = (\bQ_2 \odot \bQ_1)\bw$, in the third equality we use the identity $(\bQ_1 * \bQ_2)\odot (\bQ_3 * \bQ_4) = (\bQ_1 \odot \bQ_3) * (\bQ_2 \odot \bQ_4)$, the third-to-last equality follows from the fact that $\bZ$ is a community membership matrix, and in the second-to-last equality the identity $(\bQ_1\bQ_2)\odot(\bQ_3\bQ_4)=(\bQ_1\otimes\bQ_3)(\bQ_2\odot\bQ_4)$. Clearly,  $\|(\bm{Z\bbeta Z}^\top*\bA)\bx \|_2^2$ is now decomposed into $K$ independent quadratic forms in terms of the rows of $\bbeta$. Thus, 
    \begin{align}
        \frac{\partial}{\partial\bbeta}\|(\bm{Z\bbeta Z}^\top*\bA)\bx \|_2^2 
        &=  2\left[\bH_1 \bbeta_{1,\cdot},\ \bH_2 \bbeta_{2, \cdot},\ \ldots,\ \bH_K \bbeta_{K,\cdot}\right]^\top , \label{eq:deri-2}
    \end{align}
    where
    \begin{align*}
        \bH_k 
        &= {\vec}^{-1} \left\{ (\bZ^{\top} \otimes \bZ^{\top})\left[(\bx\one_n^\top *\bA) \odot (\bx\one_n^\top *\bA)\right] \bZ_{., k}\right\} \\
        &= \bZ^\top {\vec}^{-1} \left\{ \left[(\bx\one_n^\top *\bA) \odot (\bx\one_n^\top *\bA)\right] \bZ_{., k}\right\} \bZ \\
        & =\bZ^{\top}\left(\bx\one_n^{\top} * \bA\right) \diag\left(\bZ_{\cdot, k}\right)\left(\bx\one_n^{\top} * \bA\right)^{\top} \bZ\\
        &= \bM_k^{\top}\bM_k 
    \end{align*}
    is a symmetric matrix for $k\in [K]$. Combining \eqref{eq:loss-decomp},\eqref{eq:deri-3} and \eqref{eq:deri-2} yields that
    \begin{align*}
        \frac{\partial L}{\partial \bbeta} &= \frac{1}{n} \left[  (\bI_{K} \odot \bbeta^\top)^\top \bH^\top - \bZ^\top(\by\bx^{\top} * \bA) \bZ \right],
    \end{align*}
    where $\bH = [\bH_1,\ldots,\bH_K]$.
    Setting the above to zero finishes the proof.
\end{proof}

\subsection{Proof of \Cref{thm:eig-Hk}}
\begin{proof}[Proof of \Cref{thm:eig-Hk}]
    Recall the decomposition of $\bH_k$ in \eqref{equ:I_1+I_2+I_3+I_4}.
    Because of \eqref{equ:S1+S2}, we further have $\EE(\bH_k) = \EE(\bS_2) +\bI_4$.
    This implies that
    \begin{align}
        \|\bH_k - \mathbb{E}(\bH_k)\| & \le \|\bI_1 - \mathbb{E}(\bI_1)\| + 2 \|\bI_2\| \le \|\bS_1\| + \|\bS_2 - \mathbb{E}(\bS_2)\| + 2 \|\bI_2\|.\label{eq:Nk-bound}
    \end{align}

    In particular, $\bS_1$ is symmetric with respective to 
    $(\bA - \mathbb{E}(\bA) )_{i, j}$ and $ \left(\bA - \mathbb{E}(\bA)\right)_{i^\prime, j}$, which allows us to bound such a matrix quadratic form by decoupling. 
    
    To proceed, let $\tilde{\bA}$ be an independent copy of $\bA$. Similarly, we have $(\tilde{\bA} - \mathbb{E}(\bA))_{i^\prime, j} = 0$ if $i^\prime = j$.
    We define the following
    \begin{align}
    & \tilde{\bI}_1 = \bZ^\top \left[ \bx \bx^\top * (\bA-\EE(\bA)) \diag(\bZ_{., k}) (\tilde{\bA}-\EE(\bA))\right]\bZ, \label{equ: I_1_tilde}\\
    & \tilde{\bS}_1 = \sum_{k_1=1}^K \sum_{k_2 =1}^K  \sum_{\psi_i = k_1} \sum_{\psi_{i^\prime} = k_2} \sum_{\psi_j = k} x_i x_{i^\prime} \ind\left(i\ne i^\prime \right)(\bA - \EE(\bA) )_{i, j} (\tilde{\bA} - \EE(\bA))_{i^\prime, j} \be_{k_1} \be_{k_2}^\top, \label{equ: 
    S_1_tilde}\\
    & \tilde{\bS}_2 = \sum_{k_1=1}^K  \sum_{\psi_i = k_1} \sum_{\psi_j = k} x_i^2  (\bA - \EE(\bA))_{i, j}  (\tilde{\bA} - \EE(\bA))_{i, j}    \be_{k_1} \be_{k_1}^\top \label{equ: S_2_tilde}. 
    \end{align}
    
    By the decoupling result in Theorem 1 of \citet{de1995decoupling}, we have 
    \begin{equation*}
        \mathbb{P}(\|\bS_1\| \ge t)\le C_1 \cdot \mathbb{P} \left(\|\tilde{\bS}_1\|\ge \frac{t}{C_1}\right),
    \end{equation*}
    for any $t> 0$ and some universal constant $C_1 > 0$. Note that $\tilde{\bI}_1 = \tilde{\bS}_1 + \tilde{\bS}_2$, leading to $\|\tilde{\bS}_1\|\le \|\tilde{\bI}_1\| + \|\tilde{\bS}_2\|$. It then follows from Lemma \ref{lem:Nk-tI1} and Lemma \ref{lem:Nk-tS2} that 
    \begin{align*}
    \|\tilde{\bS}_1\| & \le \alpha(\delta) (\|\bx\|^2 + \|\bx^{(k)}\|^2 )^{1/2}\left( (\sum_{k_2=1}^K \|\bx^{(k_2)}\|_\infty^2)^{1/2}+n^{-1/2}\|\bx\|\right) s_n (n_k n)^{1/2} \log n \\
    & + \frac{4 \times 6^{1/2}}{3} \max_{k^\prime \in [K]}\|\bx^{(k^\prime)}\|_4^2(n_k\log n)^{1/2}s_n+  4\|\bx\|_{\infty}^2 \log n, 
    \end{align*}
    with probability at least $1 - {2K(n_k+2)}/{n^2}$, where $\bx^{(k^\prime)} = \bZ_{., k^\prime} * \bx$ for any $k^\prime \in [K]$ and $\alpha(\delta)= (8\delta + 4(4\delta^2 + 9)^{1/2}/3$. Hence, with probability at least $1 - 2C_1K(n_k +2)/n^2$, we have
    \begin{align*}
    \|\bS_1\| & \le \alpha(\delta)C_1 (\|\bx\|^2 + \|\bx^{(k)}\|^2 )^{1/2}\left( (\sum_{k_2=1}^K \|\bx^{(k_2)}\|_\infty^2)^{1/2}+n^{-1/2}\|\bx\|\right) s_n (n_k n)^{1/2} \log n \\
    & + \frac{4\times 6^{1/2}}{3}C_1 \max_{k^\prime\in [K]}\|\bx^{(k^\prime)}\|_4^2(n_k\log n)^{1/2}s_n+  4C_1\|\bx\|_{\infty}^2 \log n, 
    \end{align*}
    where the first term dominates. 
    
   By the probabilistic upper bound for $\|\bS_2 - \EE(\bS_2) \|$ in \Cref{lem:Nk-S2} and that for $\bI_2$ in \Cref{lem:Nk-I2}, the above upper bound and \eqref{eq:Nk-bound} yields that
    \begin{align*}
       & \quad \|\bH_k - \EE(\bH_k)\| \\
        &\leq \alpha(\delta)C_1 ( \|\bx\|^2 + \|\bx^{(k)}\|^2 )^{1/2}\left((\sum_{k_2=1}^K \|\bx^{(k_2)}\|_\infty^2)^{1/2}+n^{-1/2}\|\bx\|\right) s_n (n_k n)^{1/2} \log n \\
    &\quad\quad  + \frac{4\times 6^{1/2}}{3}C_1 \max_{k^\prime\in [K]}\|\bx^{(k^\prime)}\|_4^2(n_k\log n)^{1/2}s_n+  4C_1\|\bx\|_{\infty}^2 \log n\\
    & \quad\quad + 2 \times 2^{1/2} \max_{k^\prime \in [K]}\|\bx^{(k^\prime)}\|_4^2(n_ks_n\log n)^{1/2}+  4\|\bx\|_{\infty}^2 \log n\\
    &\quad \quad +2 \tilde{\alpha}(\delta) s_n \{( \|\bx\|^2 + \|\bx^{(k)}\|^2 )n_k s_n \sum_{k_2=1}^K \|\bx^{(k_2)}\|_1^2 \log n \}^{1/2}\\
    & \le  \alpha(\delta)(C_1 + s_n^{1/2}) ( \|\bx\|^2 + \|\bx^{(k)}\|^2 )^{1/2}\left((\sum_{k_2=1}^K \|\bx^{(k_2)}\|_\infty^2)^{1/2}+n^{-1/2}\|\bx\|\right) s_n (n_k n)^{1/2} \log n \\
    & \quad\quad + 2\times 2^{1/2} \left(\frac{2\times 3^{1/2}}{3}C_1 s_n^{1/2} + 1\right)  \max_{k^\prime\in [K]}\|\bx^{(k^\prime)}\|_4^2(n_ks_n\log n)^{1/2}+  4(C_1+1) \|\bx\|_{\infty}^2 \log n,
    \end{align*}
    with probability at least $1 - 2K(C_1n_k +2C_1 + 2)/n^2$, where $\tilde{\alpha}(\delta) = \{2\delta + 2(\delta^2 + 9)^{1/2}\}/3 <\alpha(\delta)/2$, and the last inequality use the fact that $\sum_{k_2=1}^K \|\bx^{(k_2)}\|_1^2\le n \sum_{k_2=1}^K \|\bx^{(k_2)}\|_\infty^2$.
\end{proof}

\subsection{Proof of \Cref{thm:consistency}}
\begin{proof}[Proof of \Cref{thm:consistency}]
    By the definition of $\bH_k$, we have 
    \begin{align}
        &\quad \hat{\bbeta}_{k, .} - \bbeta_{k, .}^*  = \bH_k^{-1} \bM_k^\top(\bM_k\beta_k^*+\epsilon) -\beta_k^* = \bH_k^{-1} \bM_k^\top\epsilon, 
        \label{eq:diff-beta}
    \end{align}
yielding that 
    \begin{align*}
     \EE(\hat{\bbeta}_{k, .} - \bbeta_{k, .}^*) =  \EE_{\bA}\left(\EE_{\epsilon}(\hat{\bbeta}_{k, .} - \bbeta_{k, .}^*) \vert \bA \right) = \EE_{\bA}\left(\bH_k^{-1} \bM_k^\top\EE_{\epsilon}(  \epsilon) \right)=  \zero_K,
    \end{align*}
    and we can conclude that $\hat{\bbeta}_{k, .}$ is an unbiased estimator of $\bbeta^*$. 
    Moreover, the estimation error can be expressed as
    \begin{align*}
    \|\hat{\bbeta}_{k, .} - \bbeta_{k, .}^* \| = \left(\sum_{k^\prime =1}^K \left(( \bH_k^{-1} \bM_k^{\top})_{k^\prime, .}^\top \epsilon \right)^2\right)^{1/2} = \left(\sum_{k^\prime =1}^K \left( \sum_{\psi_j = k^\prime}( \bH_k^{-1} \bM_k^{\top})_{k^\prime, j} \epsilon_j \right)^2\right)^{1/2}. 
    \end{align*}
    Given $\bx$ and $\bA$, and thus $\bH_k$ and $\bM_k$, the $\epsilon_j$'s are independent sub-exponential random variables with parameter $(\sigma_{\epsilon}^2, b_{\epsilon})$, we have $\sum_{\psi_j = k^\prime}( \bH_k^{-1} \bM_k^{\top})_{k^\prime, j} \epsilon_j$ is a sub-exponential variable with parameter $(\|(\bH_k^{-1}\bM_k^\top)_{k^\prime, .}\|^2 \sigma_{\epsilon}^2$,
    $\|\bH_k^{-1}\bM_k^{\top}\|_{\max} b_{\epsilon})$. It then follows from Bernstein's inequality that with probability at least $1 - 2/n^2$
    \begin{align*}
    \vert \sum_{\psi_j = k^\prime}( \bH_k^{-1} \bM_k^{\top})_{k^\prime, j} \epsilon_j \vert &\le \left(2\|(\bH_k^{-1}\bM_k^\top)_{k^\prime, .}\|^2 \sigma_{\epsilon}^2 \log n^2\right)^{1/2} + \|\bH_k^{-1}\bM_k^{\top}\|_{\max} b_{\epsilon}\log n^2\\
    & =  2\|(\bH_k^{-1}\bM_k^\top)_{k^\prime, .}\| \sigma_{\epsilon}\left( \log n\right)^{1/2} + 2\|\bH_k^{-1}\bM_k^{\top}\|_{\max} b_{\epsilon}\log n, 
    \end{align*}
    for any $k^\prime \in [K]$. It then follows from the union bound that with probability at least $1 - (2K)/n^2$ that 
    \begin{align}
    \label{equ: epsilon_error}
     \|\hat{\bbeta}_{k, .} - \bbeta_{k, .}^* \| & \le \left(\sum_{k^\prime =1}^K  8\|(\bH_k^{-1}\bM_k^\top)_{k^\prime, .}\|^2 \sigma_{\epsilon}^2 \log n   +  8K\|\bH_k^{-1}\bM_k^{\top}\|_{\max}^2 b_{\epsilon}^2 (\log n)^2 \right)^{1/2} \nonumber\\
     & \le 2\left( 2\|(\bH_k^{-1}\bM_k^\top)\|_F^2 \sigma_{\epsilon}^2 \log n   +  2K\|\bH_k^{-1}\bM_k^{\top}\|_{\max}^2 b_{\epsilon}^2 (\log n)^2 \right)^{1/2} \nonumber\\
    & \le 2  \|(\bH_k^{-1}\bM_k^\top)\|_F \sigma_{\epsilon} ( 2\log n)^{1/2} + 2\|\bH_k^{-1}\bM_k^{\top}\|_{\max} b_{\epsilon} (2K)^{1/2}  \log n \nonumber\\
    & \le 2 \left( 2K /\lambda_{\min}(\bH_k)\right)^{1/2}(\sigma_{\epsilon}+ b_{\epsilon})\log n.
    \end{align}

    We next proceed to lower bound the smallest eigenvalue of $\bH_k$. By Weyl's inequality, we have
    \begin{equation*}
    \lambda_{\min}(\bH_k)\ge \lambda_{\min}(\EE(\bH_k)) - \|\bH_k - \EE(\bH_k)\|. 
    \end{equation*}
    According to the decomposition of $\bH_k$, we have 
    \begin{equation*}
    \lambda_{\min}(\EE(\bH_k)) = \lambda_{\min}(\EE(\bI_1 + \bI_4]) =  \lambda_{\min}(\EE(\bS_2 + \bI_4)) \ge \lambda_{\min}(\EE(\bS_2)), 
    \end{equation*}
    where the inequality comes from the fact that both $\EE(\bS_2)$ and $\EE(\bI_4)$ are positive semi-definite matrices. Note that $\bS_2$ is a diagonal matrix and $\EE(\bS_2)$ is given in \eqref{equ: ES_2}. Therefore, 
    \begin{align*}
   & \quad \lambda_{\min}(\EE(\bS_2)) \nonumber\\
   & = \min \left\{\min_{k_1 \ne k } \sum_{\psi_i = k_1} \sum_{\psi_j = k}x_i^2 \EE(A_{i, j})(1 - \EE(A_{i, j})), \sum_{\psi_i = \psi_j = k, i< j}(x_i^2+x_j^2)\EE(A_{i, j})(1 - \EE(A_{i, j}))    \right\} \nonumber\\
   & = \min \left\{\min_{k_1 \ne k } B_{k_1, k}(1- B_{k_1, k})n_k\|\bx^{(k_1)}\|^2, B_{k, k}(1 - B_{k, k})(n_k - 1)\|\bx^{(k)}\|^2    \right\} \nonumber\\
   &  \ge \min_{k^\prime \in [K]}   B_{k^{\prime}, k}(1 - B_{k^{\prime}, k})(n_k - 1)\|\bx^{(k^{\prime})}\|^2. 
    \end{align*}
    According to \Cref{thm:eig-Hk},
   \begin{align}
    \label{equ: lambda_min_prob}
   & \quad \lambda_{\min}(\bH_k)\\
   & \ge \min_{k^\prime \in [K]}   B_{k^{\prime}, k}(1 - B_{k^{\prime}, k})(n_k - 1)\|\bx^{(k^{\prime})}\|^2 \nonumber\\
    & -  \alpha(\delta)(C + s_n^{1/2}) (\|\bx\|^2 + \|\bx^{(k)}\|^2 )^{1/2}\left\{ \left(\sum_{k_2=1}^K \|\bx^{(k_2)}\|_\infty^2\right)^{1/2}+ n^{-1/2}\|\bx\|\right\} s_n (n_k n)^{1/2} \log n \nonumber\\
    & - 2\times2^{1/2} \left(\frac{2\times 3^{1/2}}{3}C_1 s_n^{1/2} + 1\right)  \max_{k^\prime\in [K]}\|\bx^{(k^\prime)}\|_4^2(n_ks_n\log n)^{1/2}+  4(C_1+1) \|\bx\|_{\infty}^2 \log n,
   \end{align}
    with probability at least $1 - 2K(Cn_k +2C + 2)/n^2$. 
    
Finally, combing \eqref{equ: epsilon_error} and \eqref{equ: lambda_min_prob}, we have
\begin{align*}
& \quad \|\hat{\bbeta}_{k, .} - \bbeta_{k, .}^* \|\\
& \le \left\{ \min_{k^\prime \in [K]}   B_{k^{\prime}, k}(1 - B_{k^{\prime}, k})(n_k - 1)\|\bx^{(k^{\prime})}\|^2 \right.\nonumber\\
    & -  \alpha(\delta)(C_1 + s_n^{1/2}) (\|\bx\|^2 + \|\bx^{(k)}\|^2 )^{1/2}\left
    \{ \left( \sum_{k_2=1}^K \|\bx^{(k_2)}\|_\infty^2\right)^{1/2}+n^{-1/2}\|\bx\|\right\} s_n (n_k n)^{1/2} \log n \nonumber\\
    & \left. - 2\times2^{1/2} \left(\frac{2\times 3^{1/2}}{3}C_1 s_n^{1/2} + 1\right)  \max_{k^\prime\in [K]}\|\bx^{(k^\prime)}\|_4^2(n_ks_n\log n)^{1/2}+  4(C_1+1) \|\bx\|_{\infty}^2 \log n \right\}^{-1/2}\\
    & \times 2 \left( 2K\right)^{1/2}(\sigma_{\epsilon}+ b_{\epsilon})\log n,
\end{align*}
with probability at least $1 - 2K(Cn_k +2C + 3)/n^2$. This completes the proof of \Cref{thm:consistency}. 
\end{proof}

\begin{proof}[Proof of \Cref{thm: mis_Z}] Denote $\hat{\bZ} = \bZ + \bE$ be the estimated community membership matrix. Since there are $\alpha_n$ vertices have been mis-clustered. There are $\alpha_n$-rows of $\bE$ are non-zeros. Moreover, for each non-zero rows in $\bE$, there are exactly one 1 and one -1, while all the other coordinates are 0. Also, $\bE_{i, k} = -1$ must implies $\bZ_{i, k} = 1$ and thus $\hat{\bZ}_{i, k} = 0$. Without loss of generality, assume the first $\alpha_n$ vertices have been mis-clustered, leading to non-zero rows only locating in the first $\alpha_n$ rows of $\bE$. 

When $\bZ$ is misspecified by $\hat{\bZ}$, we build a neighborhood regression model as
\begin{align*}
\by &= (\bZ \bbeta \bZ^{\top} * \bA)\bx + \bepsilon\\
& = \left(\hat{\bZ} - \bE) \bbeta (\hat{\bZ} -\bE)^{\top} * \bA\right)\bx + \bepsilon\\
& = \left(\hat{\bZ} \bbeta \hat{\bZ}^{^\top} * \bA\right)\bx + \left[\left( \bE \bbeta \bE^\top - \hat{\bZ} \bbeta \bE^\top - \bE \bbeta \hat{\bZ}^\top\right) * \bA\right]\bx+ \bepsilon\\
& : =  \left(\hat{\bZ} \bbeta \hat{\bZ}^{^\top} * \bA\right)\bx + \tilde{\bepsilon}
\end{align*}
where $\tilde{\bepsilon} = \left[\left( \bE \bbeta \bE^\top - \hat{\bZ} \bbeta \bE^\top - \bE \bbeta \hat{\bZ}^\top\right) * \bA\right]\bx+ \bepsilon$. Denote $\btau =\bE \bbeta \bE^\top - \hat{\bZ} \bbeta \bE^\top - \bE \bbeta \hat{\bZ}^\top$. It is clear that $\|\btau\|_{\infty} \le 8 \|\bbeta\|_{\infty}$, and $\tau_{i, j} = 0$ if $i>\alpha_n $ and $j> \alpha_n$. For notation simplicity, we denote $\bG = \bH^{-1}_k \bM_k$. It then follows from the proof of \Cref{thm:consistency} that 
\begin{align*}
    \|\hat{\bbeta}_{k, .}^{\mis} - \bbeta_{k, .}^* \| = \left(\sum_{k^\prime =1}^K \left( \sum_{\psi_j = k^\prime} G_{k, k^\prime, j} \tilde{\epsilon}_j \right)^2\right)^{1/2}. 
\end{align*}
Moreover, 
\begin{align*}
\sum_{\psi_j  = k^\prime} G_{k^\prime, j} \tilde{\epsilon}_j & =
\sum_{\psi_j = k^\prime} G_{k^\prime, j} \epsilon_j + 
\sum_{\psi_j = k^\prime} G_{k^\prime, j}\left[(\btau * \bA)\bx\right]_j\\
& =
\sum_{\psi_j = k^\prime} G_{k^\prime, j} \epsilon_j + \sum_{i=1}^n \sum_{\psi_j = k^\prime}G_{k^\prime, j} \tau_{i, j} A_{i, j}x_i
\end{align*}
Under the event in \Cref{assum: x}, condition on $\bA$, $x_i$'s are sub-Gaussian random variables, with variance proxy $\gamma^2 \log n $. Therefore, with probability at least $1 - \kappa_n - \frac{2}{n^2}$, 
\begin{align*}
|\sum_{i=1}^n \sum_{\psi_j = k^\prime}G_{k^\prime, j} \tau_{i, j} A_{i, j}x_i| \le 2\gamma \sqrt{\sum_{i=1}^n (\sum_{\psi_j =k^\prime} G_{k^\prime, j} \tau_{i, j} A_{i, j})^2\  } (\log n)^{3/2} \le 2 \gamma \|\btau * \bA\| \|\bG_{k^\prime, .}\| (\log n)^{3/2}.
\end{align*}
By the matrix Bernstein's inequality, $\|\btau * \bA\|$ is concentration around $\| \btau * \bP\| = O_p(s_n \sqrt{n \alpha_n})$. 
Therefore, by the triangle inequality and the proof of \Cref{thm:consistency}, 
there exists some universal constant $\tilde{C}_4$ and $\tilde{C}_5$, such that with probability at least $1 - \kappa_n - \tilde{C}_4 K/n^2$, we have 
\begin{align*}
\|\hat{\bbeta}_{k, .}^{\mis} - \bbeta_{k, .}^* \| & \le 2 \left( 2K /\lambda_{\min}(\bH_k)\right)^{1/2}(\sigma_{\epsilon}+ \tilde{C}_5 b_{\epsilon})\log n + \gamma s_n \sqrt{n\alpha_n} \log^2 n \|\bG\|_F \\
& \le 2 \left( 2K /\lambda_{\min}(\bH_k)\right)^{1/2}(\sigma_{\epsilon}+ b_{\epsilon})\log n + \tilde{C}_5\gamma s_n \sqrt{n\alpha_n} \log^2 (K /\lambda_{\min}(\bH_k))^{1/2}   \\
& \le \max\{2\sqrt{2}, \tilde{C}_5\} \left(K /\lambda_{\min}(\bH_k)\right)^{1/2}(\sigma_{\epsilon}+ b_{\epsilon} + \gamma s_n \sqrt{n \alpha_n } \log n)\log n. 
\end{align*}
By the result in \Cref{cor: root_n}, we have 
\begin{align*}
     \|\hat{\bbeta}_{k, .}^{\mis} - \bbeta_{k, .}^* \| \le C_5 \frac{ K^{1/2}(\sigma_{\epsilon} + b_{\epsilon} + \gamma s_n \sqrt{n \alpha_n} \log n)\log n }{(s_n n_k n_{\min})^{1/2}}. 
\end{align*}
with probability at least $1 - \tilde{C_4}/n - \kappa_n$, for some constant $C_5$. 
\end{proof}

\subsection{Proof of \Cref{thm:BLUE}}
\begin{proof}[Proof of \Cref{thm:BLUE}]
Recall that $\hat{\bbeta}_{k,\cdot} = \bH_k^{-1}\bM_k^{\top}\by$ and $\by = \sum_{k^\prime=1}^K \bM_{k^\prime} \bbeta_{k^\prime, .}^* + \bepsilon$. Any estimator of $\bbeta_{k, .}^*$ that is linear in $\by$ takes the form $\bar{\bbeta}_k = \hat{\bbeta}_{k, .} + \tilde{\bM}\by = (\bH_k^{-1}\bM_k^\top + \tilde{\bM})\by$, for $\tilde{\bM}\in\RR^{K\times n}$.
The unbiased property of $\hat{\bbeta}_{k,\cdot}$ and $\bar{\bbeta}_k$, together with the zero-mean assumption of $\bepsilon$, implies that 
\begin{align*}
\EE(\bar{\bbeta}_k) = \bbeta^*_{k, .} + \sum_{k^\prime=1}^K \tilde{\bM} \bM_{k^\prime}\bbeta^*_{k^\prime, .} =\bbeta^*_{k,.}, 
\end{align*}
where the expectation is taken with respect to $\by$, leading to $\sum_{k^\prime=1}^K \tilde{\bM} \bM_{k^\prime}\bbeta^*_{k^\prime, .}=\zero$. 
Since $\bbeta^*$ is unknown and could be arbitrary, we have $\tilde{\bM}\bM_{k^\prime}=\zero$, for any $k^\prime \in [K]$. 

Under the assumption that the variance of $\epsilon_i$'s are the same within community $k$, denoted by $\sigma
_k^2$, we can decomposed the covariance matrix of $\bar{\bbeta}_{k, .}$ as 
\begin{align*}
 \Var( \bar{\bbeta}_k )  & = \Var( \hat{\bbeta}_{k,\cdot} ) + \sigma^2_k \bH_k^{-1} \bM_k^\top \tilde{\bM}^\top  + \sigma^2_k \tilde{\bM}\bM_k\bH_k^{-1} + \sigma^2_k\tilde{\bM}\tilde{\bM}^\top \\
& = \Var( \hat{\bbeta}_{k,\cdot} ) + \sigma^2_k\tilde{\bM}\tilde{\bM}^\top 
\end{align*}
where the first equality follows from the fact that the $i$-th row of $\bM_k$ is $\zero$ if $\psi_i \ne k$, and the second one follows from $\tilde{\bM}\bM_k =\zero$. Finally, since $\sigma^2_k\tilde{\bM}\tilde{\bM}^\top$ is positive semi-definite, we conclude that 
\begin{align*}
\Var( \hat{\bbeta}_{k,\cdot} ) \preceq \Var( \bar{\bbeta}_k ). 
\end{align*}
\end{proof}

\subsection{Proof of \Cref{thm:minimax}}

\begin{proof}[Proof of \Cref{thm:minimax}]

We split the proof into two parts.

\paragraph{Part (1) The lower bound.}
For any $\lambda>0$, define the ridge regression estimator for the $k$th sub-problem as 
\begin{align*}
\hat{\bbeta}_{k, \lambda} = \arg \min_{\bbeta \in \RR^K}\left\{ \frac{1}{n_k} \| \by * \bZ_{., k} - \bM_k \bbeta_{k, .}\|_2^2 + \lambda\|\bbeta\|^2 \right\} = \frac{1}{n_k} \left(\frac{1}{n_k} \bH_k + \lambda \bI\right)^{-1}\bM_{k}^{\top} (\by * \bZ_{., k}). 
\end{align*}
Note that 
\begin{align*}
\by * \bZ_{., k} = (\bM_{k, .} \bbeta_{k, .}^* + \bepsilon) * \bZ
_{., k} = \bM_{k, .} \bbeta_{k, .}^* * \bZ_{., k} + \bepsilon * \bZ
_{., k} = \bM_{k, .} \bbeta_{k, .}^* + \bepsilon * \bZ
_{., k}, 
\end{align*}
where the last equality follows from the fact that the $i$th row of $\bM_{k, .}$ is $\zero$ if $\psi_i \ne k $. Moreover, 
\begin{align*}
\bM_{k, .}^\top (\by * \bZ_{., k} ) = \bH_k \bbeta_{k, .}^* + \bM_{k, .}^\top \bepsilon. 
\end{align*}
As such, 
\begin{align}
\label{equ: calE_decom}
\EE\left(\mathcal{E}(\hat{\bbeta}_{k, \lambda})\right)  & = \EE \left( \left\|\frac{1}{n_k} \left(\frac{1}{n_k} \bH_k + \lambda \bI\right)^{-1}\left( \bH_k \bbeta_{k, .}^* + \bM_{k, .}^\top \bepsilon\right) -\bbeta_{k, .}^* \right\|_{\bSigma}^2\right) \nonumber\\
 & = \EE\left(\left\|\frac{1}{n_k} \left(\frac{1}{n_k} \bH_k + \lambda \bI\right)^{-1}  \bM_{k, .}^\top \bepsilon   - \lambda \left(\frac{1}{n_k} \bH_k + \lambda \bI\right)^{-1}\bbeta_{k, .}^* \right\|_{\bSigma}^2\right) \nonumber\\
 & =\EE\left(\EE\left(\left\|\frac{1}{n_k}  \bM_{k, .}^\top \bepsilon   - \lambda\bbeta_{k, .}^* \right\|_{\left(\frac{1}{n_k} \bH_k + \lambda \bI\right)^{-1} \bSigma \left(\frac{1}{n_k} \bH_k + \lambda \bI\right)^{-1} }^2 \Big\vert \bM_{k, .}\right)\right) \nonumber\\
 & = \lambda^2\EE\left( \|\bbeta_{k, .}^* \|_{\bS}^2\right) + \frac{1}{n_k^2}\sum_{\psi_i = k}\Var(\epsilon_i) \EE\left(\|(\bM_{k, .})_{i, .}\|^2_{\bS}\right), 
\end{align}
where $\bS = \left(n_k^{-1} \bH_k + \lambda \bI\right)^{-1} \bSigma \left(n_k^{-1} \bH_k + \lambda \bI\right)^{-1}$, and the last equality can be obtained by expanding.

Let 
\begin{align*}
\mathcal{P}_{k, \text{Gauss}}(P_X, P_{\bA}, \sigma_{\max}) = \left\{ P_{X, \bA, Y}: \bx \sim P_X, \bA \sim P_{\bA}, \bbeta^*_{k,.} \in \RR^K, \bepsilon \sim \mathcal{N}(\zero, \sigma_{\max}^2 \bI) \right\}, 
\end{align*}
Since $ \mathcal{P}_{k, \text{Gauss}}(P_X, P_{\bA}, \bsigma^2) \subset \mathcal{P}_k(P_X, P_{\bA}, \bsigma^2)$, we have
\begin{align*}
\inf_{\hat{\bbeta}_{k, .}} \sup_{\bbeta^*_{k, .} \in \mathcal{P}_k(P_X, P_{\bA}, \bsigma^2)}\EE\left(\mathcal{E}(\hat{\bbeta}_{k, .})\right) \ge \inf_{\hat{\bbeta}_{k, .}} \sup_{\bbeta^*_{k, .} \in \mathcal{P}_{k, \text{Gauss}}(P_X, P_{\bA}, \bsigma^2)}\EE\left(\mathcal{E}(\hat{\bbeta}_{k, .})\right).
\end{align*}
It thus suffices to derive a lower bound for the minimax risk over $\mathcal{P}_{k, \text{Gauss}}$. 

Under the distribution family $\mathcal{P}_{k, \text{Gauss}}$, we have $\Var(\epsilon_i) = \sigma_{\max}^2$. Consider a prior of $\bbeta^*_{k, .}$ $Q_{\lambda} = \mathcal{N}(\zero_K, \sigma_{\max}^2/(\lambda n_k)\bI_K)$, the posterior reads
\begin{align*}
& \quad \Pi_{\psi_i = k} \frac{1}{(2\pi)^{1/2} \sigma_{\max}}\exp\left\{- \frac{(y_i - (\bM_{k})_{i, .}^\top \bbeta_{k, .}^*)^2}{2\sigma_{\max}^2}\right\} \times  \left(\frac{\lambda n_k}{ 2\pi\sigma^2_{\max}}\right)^{K/2} \exp\{- \frac{\lambda n_k (\bbeta_{k, .}^*)^\top \bbeta_{k, .}^*}{2\sigma_{\max}^2}\}\\
& \propto \exp \left\{-\frac{1}{2\sigma_{\max}^2}\left[\sum_{\psi_i = k} \left(y_i - (\bM_{k})_{i, .}^\top \bbeta_{k, .}^*\right)^2 + \lambda n_k(\bbeta_{k, .}^*)^\top \bbeta_{k, .}^*\right] \right\}\\
& \propto \exp \left\{-\frac{1}{2\sigma_{\max}^2}\left[ (\bbeta_{k, .}^*)^\top (\bH_k +  \lambda n_k \bI) \bbeta_{k, .}^* - 2\by^\top \bM_{k} \bbeta_{k, .}^*\right] \right\}\\
& \propto \exp \left\{-\frac{1}{2\sigma_{\max}^2}\left[ (\bbeta_{k, .}^* - \hat{\bbeta}_{k, \lambda})^\top (\bH_k +  \lambda n_k \bI) (\bbeta_{k, .}^* - \hat{\bbeta}_{k, \lambda}) \right] \right\}. 
\end{align*}
Herein, the proportion notation $\propto$ drops some factors that do not depend on $\bbeta_{k, .}^*$. This shows that the posterior of $\bbeta_{k, .}^*$ follows $\mathcal{N}\big(\hat{\bbeta}_{k, \lambda}, \sigma_{\max}^2 (\bH_k + \lambda n_k \bI)^{-1}\big)$. Therefore, maximizing the posterior likelihood yields that $\widehat{\bbeta_{k, .}^*} = \hat{\bbeta}_{k, \lambda}$, leading to 
\begin{align*}
\inf_{\hat{\bbeta}_{k, .}} \sup_{\bbeta^*_{k, .} \in \mathcal{P}_{k, \text{Gauss}}(P_X, P_{\bA}, \bsigma^2)}\EE\left(\mathcal{E}(\hat{\bbeta}_{k, .})\right) & \ge \inf_{\hat{\bbeta}_{k, .}} \EE_{\bbeta^*_{k, . \sim Q_{\lambda}}}\EE\left(\mathcal{E}(\hat{\bbeta}_{k, .})\right)\ge \EE_{\bbeta^*_{k, . \sim Q_{\lambda}}}\EE\left(\mathcal{E}(\hat{\bbeta}_{k, \lambda})\right).
\end{align*}
Moreover, by Fubini's theorem,
\begin{align*}
& \quad \EE_{\bbeta^*_{k, . \sim Q_{\lambda}}}\EE\left(\mathcal{E}(\hat{\bbeta}_{k, \lambda})\right)\\
&= \EE_{\bbeta^*_{k, . \sim Q_{\lambda}}}\left(\lambda^2\EE\left( \|\bbeta_{k, .}^* \|_{\bS}^2\right) + \frac{1}{n_k^2}\sum_{\psi_i = k}\Var(\epsilon_i) \EE\left(\|(\bM_{k, .})_{i, .}\|^2_{\bS}\right)\right)\\
& = \lambda^2 \EE\left( \EE_{\bbeta^*_{k, . \sim Q_{\lambda}}}  \left[ \|\bbeta_{k, .}^* \|_{\bS}^2\right]      \right) + \frac{\sigma_{\max}^2}{n_k^2}\sum_{\psi_i = k}\EE\left(\|(\bM_{k, .})_{i, .}\|^2_{\bS}\right)\\
& = \lambda^2 \EE\left( \EE_{\bbeta^*_{k, . \sim Q_{\lambda}}}  \left( \tr\big(\bbeta_{k, .}^*(\bbeta_{k, .}^*)^\top \bS\big)\right)    \right) + \frac{\sigma_{\max}^2}{n_k^2} \EE\left(\tr \big(\sum_{\psi_i = k}(\bM_{k, .})_{i, .} (\bM_{k, .})_{i, .}^\top \bS\big)\right)\\
& = \lambda^2 \EE\left( \frac{\sigma_{\max}^2}{\lambda n_k} \tr\big(\bS\big)\right)+ \frac{\sigma_{\max}^2}{n_k^2} \EE\left(\tr \big(\bH_k \bS\big)\right)\\
& =\frac{\sigma_{\max}^2}{n_k}  \EE\left( \tr\left(\left(\frac{1}{n_k} \bH_k + \lambda \bI\right)\bS\right)\right)\\
& =\frac{\sigma_{\max}^2}{n_k}  \tr\left( \bSigma \EE\left(\left(\frac{1}{n_k} \bH_k + \lambda \bI\right)^{-1}\right) \right). 
\end{align*}
Finally, since  the function $\lambda\mapsto\tr\left(\bSigma \EE\left(\left(n_k^{-1} \bH_k + \lambda \bI\right)^{-1}\right)\right)$ is positive and decreasing in $\lambda$, the monotone convergence theorem yields that 
\begin{align*}
\lim_{\lambda\rightarrow 0}\EE_{\bbeta^*_{k, . \sim Q_{\lambda}}}\EE\left(\mathcal{E}(\hat{\bbeta}_{k, \lambda})\right) =\sigma_{\max}^2 \tr\left(\bSigma \EE\left( \bH_k^{-1}\right) \right).
\end{align*}
Joining the pieces together, we conclude that 
\begin{align*}
\inf_{\hat{\bbeta}_{k, .}} \sup_{\bbeta^*_{k, .} \in \mathcal{P}_k(P_X, P_{\bA}, \bsigma^2)}\EE\left(\mathcal{E}(\hat{\bbeta}_{k, .})\right) \ge\sigma_{\max}^2\tr\left(\bSigma\EE (\bH_k^{-1}) \right).
\end{align*}

\paragraph{Part (2) The upper bound.} 
Taking $\lambda =  0$ in the decomposition \eqref{equ: calE_decom}, we have
\begin{align*}
\EE\left(\mathcal{E}(\hat{\bbeta}_{k, \lambda})\right) \le \frac{1}{n_k^2}\sum_{\psi_i = k}\Var(\epsilon_i) \EE\left(\|(\bM_{k, .})_{i, .}\|^2_{\bS}\right) \le\sigma_{\max}^2\tr\left(\bSigma\EE (\bH_k^{-1}) \right)
\end{align*}
Combining the lower bound in Part (1) yields the desired result.

\end{proof}

\section{Proof of supporting lemmas}\label{app:lemmas}
\subsection{Proof of \Cref{lem:Nk-tI1}}
\begin{lemma}[Bounding $\tilde{\bI}_1$]\label{lem:Nk-tI1}
    Under the same assumptions as in \Cref{thm:consistency} with $\tilde{\bI}_1$ defined in \eqref{equ: I_1_tilde}, it holds that \begin{align*}
        \|\tilde{\bI}_1\| \le \alpha(\delta) ( \|\bx\|^2 + \|\bx^{(k)}\|^2 )^{1/2}\left\{  \left(\sum_{k_2=1}^K \|\bx^{(k_2)}\|_\infty^2\right)^{1/2}+ n^{-1/2}\|\bx\|\right\} s_n (n_k n)^{1/2} \log n. 
    \end{align*}    
    with probability at least $1 - {2K(n_k+1)}/{n^2}$, where $\bx^{(k^\prime)} = \bZ_{., k^\prime} * \bx$ for any $k^\prime \in [K]$ and $\alpha(\delta)= \{8\delta + 4(4\delta^2 + 9)^{1/2}\}/3$.
\end{lemma}
\begin{proof}[Proof of \Cref{lem:Nk-tI1}]
    The idea of the proof is to employ the matrix Bernstein inequality for $\tilde{\bI}_1$ conditioned on $\tilde{\bA}$. To do this,  we divide the proof into several steps.

    \paragraph{Step 1. Decomposition of $\tilde{\bI}_1$.}
    Given $\tilde{\bA}$, we rewrite $\tilde{\bI}_1$ into the summation of a series of centered independent random matrices. Specifically, 
    \begin{align*}
    \tilde{\bI}_1 & =\sum_{k_1=1}^K \sum_{k_2 =1}^K  \sum_{\psi_i = k_1} \sum_{\psi_{i^\prime} = k_2} \sum_{\psi_j = k} x_i x_{i^\prime} (\bA - \EE(\bA) )_{i, j} (\tilde{\bA} - \EE(\bA))_{i^\prime, j} \be_{k_1} \be_{k_2}^\top\\
    & =\sum_{k_1=1, k_1 \ne k}^K   \sum_{\psi_i = k_1} \sum_{\psi_j = k} (\bA - \EE(\bA) )_{i, j} \sum_{k_2 =1}^K \sum_{\psi_{i^\prime} = k_2} x_i x_{i^\prime}(\tilde{\bA} - \EE(\bA))_{i^\prime, j} \be_{k_1} \be_{k_2}^\top\\
    & \quad + \sum_{\psi_i = \psi_j = k, i< j} (\bA - \EE(\bA) )_{i, j} \sum_{k_2 =1}^K \sum_{\psi_{i^\prime} = k_2} \left[x_i x_{i^\prime}(\tilde{\bA} - \EE(\bA))_{i^\prime, j} + x_j x_{i^\prime}(\tilde{\bA} - \EE(\bA))_{i^\prime, i} \right] \be_k \be_{k_2}^\top\\
    & := \sum_{k_1=1, k_1 \ne k}^K   \sum_{\psi_i = k_1} \sum_{\psi_j = k} \bY_{k_1, i, j} +  \sum_{\psi_i = \psi_j = k, i< j}\bY_{k, i, j}, 
    \end{align*}
    where 
    \begin{align*}
     & \bY_{k_1, i, j} = (\bA - \EE(\bA) )_{i, j} \sum_{k_2 =1}^K \sum_{\psi_{i^\prime} = k_2} x_i x_{i^\prime}(\tilde{\bA} - \EE(\bA))_{i^\prime, j} \be_{k_1} \be_{k_2}^\top, \text{ for } k_1 \ne k, \text{ and }\\
    & \bY_{k, i, j} = (\bA - \EE(\bA) )_{i, j} \sum_{k_2 =1}^K \sum_{\psi_{i^\prime} = k_2} \left[x_i x_{i^\prime}(\tilde{\bA} - \EE(\bA))_{i^\prime, j} + x_j x_{i^\prime}(\tilde{\bA} - \EE(\bA))_{i^\prime, i} \right] \be_k \be_{k_2}^\top.
    \end{align*}
    Clearly,  $\bDelta_k = \{\bY_{k_1, i, j}: \psi_i=k_1, \psi_j = k, k_1 \ne k \} \cup \{\bY_{k, i, j}:\psi_i = \psi_j = k, i<j \}$ is a set of zero-mean independent random matrices given $\tilde{\bA}$.
    
    \paragraph{Step 2. Uniform bound of spectral norms.} We provide a uniform upper bound for the spectral norms of the matrices in the set $\bDelta$. For $k_1 \ne k$,
    \begin{align*}
        \|\bY_{k_1, i, j}\| & = |(\bA - \EE(\bA) )_{i, j}| \left\| \sum_{k_2 =1}^K \sum_{\psi_{i^\prime} = k_2} x_i x_{i^\prime}(\tilde{\bA} - \EE(\bA))_{i^\prime, j} \be_{k_1} \be_{k_2}^\top\right\|\\
        & \le \|\bx\|_{\infty} \left\{\sum_{k_2 =1}^K \left(\sum_{\psi_{i^\prime} = k_2}x_{i^\prime}(\tilde{\bA} - \EE(\bA))_{i^\prime, j}\right)^2 \right\}^{1/2}. 
    \end{align*}
    Note that $\sum_{\psi_{i^\prime} = k_2} x_{i^\prime} (\tilde{\bA} - \EE(\bA))_{i^\prime, j}$ is a summation of $n_{k_2}$ independent sub-exponential random variables with parameters $\{(x_{i^{\prime}}^2 s_n, |x_{i^\prime}|): \psi_{i^\prime} = k_2\}$.  Thus, $\sum_{\psi_{i^\prime} = k_2} x_{i^\prime} (\tilde{\bA} - \EE(\bA))_{i^\prime, j}$ is sub-exponential with parameter $(\sum_{\psi_{i^\prime} = k_2} x_{i^\prime}^2 s_n,  \|\bx^{(k_2)}\|_{\infty})$. Herein, we can safely regard $0$ as a $(x_{i^\prime}^2 s_n, |x_{i^\prime}|)$-sub-exponential random variable in case $i^\prime = j$. By the conventional Bernstein inequality, we have 
    \begin{equation*}
        \PP \left(   \left|\sum_{\psi_{i^\prime} = k_2}x_{i^\prime}(\tilde{\bA} - \EE(\bA))_{i^\prime, j} \right| >  \left(2 \|\bx^{(k_2)}\|^2  s_n \log n^2 \right)^{1/2} + \|\bx^{(k_2)}\|_{\infty} \log n^2   \right) \le  \frac{2}{n^2}, 
    \end{equation*}
Together with the assumption that $s_n \ge \log n/n$, we have, with probability at least $1- 2/n^2$, 
\begin{align}
\label{equ: first_concentration_bound}
 \left|\sum_{\psi_{i^\prime} = k_2}x_{i^\prime}(\tilde{\bA} - \EE(\bA))_{i^\prime, j} \right| &\le  \left(2 \|\bx^{(k_2)}\|^2  s_n \log n^2 \right)^{1/2} + \|\bx^{(k_2)}\|_{\infty} \log n^2  \nonumber\\
 & \le 2 \left(\|\bx^{(k_2)}\|_{\infty}+n^{-1/2}\|\bx^{(k_2)}\| \right)(n s_n\log n)^{1/2}, 
\end{align}
for any given $j$ such that $\psi_j = k$ and $k_2 \in [K]$. Furthermore, according to the union bound, we have \eqref{equ: first_concentration_bound} holds simultaneously for all $j$ such that $\psi_j = k$ and $k_2\in [K]$, with probability at least $1- 2Kn_k/n^2$. We denote this union event as $E_1$. As such, $E_1$ holds with probability at least $1- 2Kn_k/n^2$. This leads to 
\begin{align*}
  \|\bY_{k_1, i, j}\| & \le 2 \|\bx\|_{\infty} \left\{n s_n\log n  \left(\sum_{k_2= 1}^K \|\bx^{(k_2)}\|_{\infty}^2+ 2n^{-1/2}\sum_{k_2 =1}^K \|\bx^{(k_2)}\|_{\infty}\|\bx^{(k_2)}\| +\frac{\|\bx\|^2}{n}\right) \right\}^{1/2}\\
  & \le 2\|\bx\|_{\infty} \left\{ \left(\sum_{k_2=1}^K \|\bx^{(k_2)}\|_\infty^2\right)^{1/2}+ n^{-1/2}\|\bx\| \right\} (ns_n \log n )^{1/2}, 
\end{align*}
under the event $E_1$. In addition, for any $i< j$ such that $\psi_i = \psi_j = k$, we have
\begin{align*}
\|\bY_{k, i, j}\| & = |(\bA - \EE(\bA) )_{i, j}| \left\| \sum_{k_2 =1}^K \sum_{\psi_{i^\prime} = k_2} \left[x_i x_{i^\prime}(\tilde{\bA} - \EE(\bA))_{i^\prime, j} + x_j x_{i^\prime}(\tilde{\bA} - \EE(\bA))_{i^\prime, i} \right] \be_k \be_{k_2}^\top\right\|\\
& \le \left[ \sum_{k_2 =1}^K \left\{\sum_{\psi_{i^\prime} = k_2}\left(x_i x_{i^\prime}(\tilde{\bA} - \EE(\bA))_{i^\prime, j} + x_j x_{i^\prime}(\tilde{\bA} - \EE(\bA))_{i^\prime, i}\right)\right\}^2 \right]^{1/2}\\ 
& \le \|\bx\|_{\infty} \left[ 2\sum_{k_2 =1}^K \left\{ \left(\sum_{\psi_{i^\prime} = k_2} x_{i^\prime}(\tilde{\bA} - \EE(\bA))_{i^\prime, j}\right)^2  + \left(\sum_{\psi_{i^\prime} = k_2}x_{i^\prime}(\tilde{\bA} - \EE(\bA))_{i^\prime, i}\right)^2\right\}\right]^{1/2}\\
& \le 4\|\bx\|_{\infty} \left\{ \left(\sum_{k_2=1}^K \|\bx^{(k_2)}\|_\infty^2\right)^{1/2}+ n^{-1/2}\|\bx\| \right\}(ns_n \log n)^{1/2}, 
\end{align*}
where the last inequality holds under the event $E_1$. As such, define
\begin{equation}
\label{equ: R_y}
    R_{\bDelta}= 4\|\bx\|_{\infty} \left\{ \left(\sum_{k_2=1}^K \|\bx^{(k_2)}\|_\infty^2\right)^{1/2}+ n^{-1/2}\|\bx\| \right\} (ns_n \log n)^{1/2}, 
\end{equation}
which can serve as a uniform upper bound for the spectral norms of the matrices in $\bDelta$, under the event $E_1$.

\paragraph{Step 3. The second moment bound.} We now turn to derive an upper bound for the summation of the second-order moments for the matrices in $\bDelta$. Under $E_1$,  we have, for $k_1\ne k$,
\begin{align*}
& \quad \EE(\bY_{k_1, i, j}\bY_{k_1, i, j}^{\top})\\
&= \EE((\bA - \EE(\bA) )_{i, j}^2) \sum_{k_2 =1}^K \sum_{\psi_{i^\prime} = k_2} \sum_{\psi_{i^{\prime\prime}} = k_2}  x_i^2 x_{i^\prime}x_{i^{\prime \prime}}(\tilde{\bA} - \EE(\bA))_{i^\prime, j}(\tilde{\bA} - \EE(\bA))_{i^{\prime\prime}, j} \be_{k_1} \be_{k_1}^\top\\
& \preceq s_n x_i^2  \sum_{k_2=1}^K \left( \sum_{\psi_{i^\prime} = k_2} x_{i^\prime} (\tilde{\bA} - \EE(\bA))_{i^\prime, j} \right)^2\be_{k_1} \be_{k_1}^\top\\
& \preceq 4 x_i^2 \left\{ \left( \sum_{k_2=1}^K \|\bx^{(k_2)}\|_\infty^2\right)^{1/2}+ n^{-1/2}\|\bx\| \right\}^2  ns_n^2 \log n \be_{k_1} \be_{k_1}^\top.   
\end{align*}
Herein, the partial order $\bM_1 \preceq \bM_2$ for two positive semi-definite matrices $\bM_1$ and $\bM_2$ means that $\bM_2 - \bM_1$ is positive semi-definite. It then follows that 
\begin{align}
  & \quad \sum_{k_1\ne k}\sum_{\psi_i = k_1} \sum_{\psi_j = k}\\
  & \EE(\bY_{k_1, i, j}\bY_{k_1, i, j}^{\top})   \preceq 4  \left\{ \left( \sum_{k_2=1}^K \|\bx^{(k_2)}\|_\infty^2\right)^{1/2}+ n^{-1/2}\|\bx\| \right\}^2 n_k ns_n^2 \log n \sum_{k_1\ne k} \|\bx^{(k_1)}\|^2  \be_{k_1} \be_{k_1}^\top \nonumber \\
    &  \preceq 4 \left\{ \left( \sum_{k_2=1}^K \|\bx^{(k_2)}\|_\infty^2\right)^{1/2}+ n^{-1/2}\|\bx\| \right\}^2 \max_{k^\prime \in [K]\backslash\{k\}} \|\bx^{(k^\prime)}\|^2n_k  ns_n^2 \log n \sum_{k_1\ne k}  \be_{k_1} \be_{k_1}^\top.   \label{equ:second_k_1_1}
\end{align}
In addition, for any $i<j$ such that $\psi_i = \psi_j = k$, we have
\begin{align*}
& \quad \EE(\bY_{k, i, j}\bY_{k, i, j}^{\top})\\
& \preceq 2 s_n  \sum_{k_2=1}^K \left[ x_i^2 \left\{\sum_{\psi_{i^\prime} = k_2} x_{i^\prime} (\tilde{\bA} - \EE(\bA))_{i^\prime, j}\right\}^2 + x_j^2\left\{  \sum_{\psi_{i^\prime} = k_2} x_{i^\prime} (\tilde{\bA} - \EE(\bA))_{i^\prime, i}\right\}^2\right] \be_k \be_k^\top\\
& \preceq  8(x_i^2 + x_j^2) \left\{ \left(\sum_{k_2=1}^K \|\bx^{(k_2)}\|_\infty^2\right)^{1/2}+n^{-1/2}\|\bx\| \right\}^2 ns_n^2 \log n \cdot \be_k \be_k^\top. 
\end{align*}
Taking the summation for all $i<j$, we have
\begin{align}
\label{equ: second_k_1}
& \quad \sum_{\psi_i = \psi_j = k, i<j}\EE(\bY_{k, i, j}\bY_{k, i, j}^{\top}) \nonumber\\
&\le 8 \left\{ \left(\sum_{k_2=1}^K \|\bx^{(k_2)}\|_\infty^2\right)^{1/2}+ n^{-1/2}\|\bx\| \right\}^2 (n_k-1) \|\bx^{(k)}\|^2 ns_n^2 \log n \cdot \be_k \be_k^\top.
\end{align}
Combing \eqref{equ:second_k_1_1} and \eqref{equ: second_k_1} yields that
\begin{align}
\label{equ: second_order_moment_left}
&\quad \left\|\sum_{k_1\ne k}\sum_{\psi_i = k_1} \sum_{\psi_j = k}\EE(\bY_{k_1, i, j}\bY_{k_1, i, j}^{\top}) +\sum_{\psi_i = \psi_j = k, i<j}\EE(\bY_{k, i, j}\bY_{k, i, j}^{\top}) \right\| \nonumber\\
& \le 4  \left\{ \left(\sum_{k_2=1}^K \|\bx^{(k_2)}\|_\infty^2\right)^{1/2}+n^{-1/2}\|\bx\| \right\}^2 \max \left\{ \max_{k^\prime \in [K]\backslash\{k\}}\|\bx^{(k^\prime)}\|^2, 2 \|\bx^{(k)}\|^2 \right\}n_k n s_n^2 \log n.
\end{align}

Similarly, for $k_1 \ne k$, we have
\begin{align*}
& \quad \EE(\bY_{k_1, i, j}^\top \bY_{k_1, i, j})\\
& \preceq s_n x_i^2  \sum_{k_2=1}^K  \sum_{k_2^\prime =1}^K \left \vert \sum_{\psi_{i^\prime} = k_2} x_{i^\prime} (\tilde{\bA} - \EE(\bA))_{i^\prime, j} \right \vert \cdot \left\vert \sum_{\psi_{i^{\prime\prime}} = k_2^\prime} x_{i^{\prime\prime}} (\tilde{\bA} - \EE(\bA))_{i^{\prime\prime}, j} \right\vert  \be_{k_2} \be_{k_2^\prime}^\top\\
& \preceq 4 x_i^2 n s_n^2 \log n  \sum_{k_2 = 1}^K \sum_{k_2^\prime = 1}^K  \left(\|\bx^{(k_2)}\|_{\infty}+n^{-1/2}\|\bx^{(k_2)}\|\right)\left(\|\bx^{(k_2^\prime)}\|_{\infty}+n^{-1/2}\|\bx^{(k_2^\prime)}\| \right)\be_{k_2} \be_{k_2^\prime}^\top, 
\end{align*}
which leads to 
\begin{align}
\label{equ: second_order_mement_2}
& \quad \sum_{k_1 \ne k } \sum_{\psi_i = k_1 } \sum_{\psi_j = k}\EE(\bY_{k_1, i, j}^\top \bY_{k_1, i, j})  \nonumber\\
& \preceq 4  (\|\bx\|^2 - \|\bx^{(k)}\|^2)n_k n s_n^2 \log n  \\
& \quad \times \sum_{k_2 = 1}^K \sum_{k_2^\prime = 1}^K \left(\|\bx^{(k_2)}\|_{\infty}+n^{-1/2}\|\bx^{(k_2)}\| \right)\left(\|\bx^{(k_2^\prime)}\|_{\infty}+n^{-1/2}\|\bx^{(k_2^\prime)}\| \right) \be_{k_2} \be_{k_2^\prime}^\top. 
\end{align}
For $i< j$ with $\psi_i = \psi_j = k$, 
\begin{align*}
& \quad \EE(\bY_{k, i, j}^{\top} \bY_{k, i, j})\\
& \preceq s_n  \sum_{k_2=1}^K \sum_{k_2^\prime=1}^K  \left[ |x_i| \left \vert \sum_{\psi_{i^\prime} = k_2} x_{i^\prime} (\tilde{\bA} - \EE(\bA))_{i^\prime, j}\right \vert + |x_j|\left\vert \sum_{\psi_{i^\prime} = k_2} x_{i^\prime} (\tilde{\bA} - \EE(\bA))_{i^\prime, i}\right\vert\right] \\
& \cdot \left[ |x_i| \left \vert \sum_{\psi_{i^{\prime\prime}} = k_2^\prime} x_{i^{\prime\prime}} (\tilde{\bA} - \EE(\bA))_{i^{\prime\prime}, j}\right \vert + |x_j|\left\vert \sum_{\psi_{i^\prime} = k_2} x_{i^{\prime\prime}} (\tilde{\bA} - \EE(\bA))_{i^{\prime\prime}, i}\right\vert\right]\be_{k_2} \be_{k_2^\prime}^\top\\
& \preceq  8 (x_i^2 + x_j^2)  ns_n^2 \log n \sum_{k_2=1}^K \sum_{k_2^\prime=1}^K  \left(\|\bx^{(k_2)}\|_{\infty}+n^{-1/2}\|\bx^{(k_2)}\|\right)\left(\|\bx^{(k_2^\prime)}\|_{\infty}+n^{-1/2}\|\bx^{(k_2^\prime)}\|\right) \be_{k_2} \be_{k_2^\prime}^\top,
\end{align*}
leading to 
\begin{align}
\label{equ: second_order_monent_2_k}
& \quad \sum_{\psi_i = \psi_j = k , i< j } \EE(\bY_{k, i, j}^{\top}\bY_{k, i, j}) \nonumber\\
& \preceq  8 \|\bx^{(k)}\|^2(n_k-1) ns_n^2 \log n \nonumber\\
& \quad \times \sum_{k_2=1}^K \sum_{k_2^\prime=1}^K \left(\|\bx^{(k_2)}\|_{\infty}+n^{-1/2}\|\bx^{(k_2)}\|\right)\left(\|\bx^{(k_2^\prime)}\|_{\infty}+n^{-1/2}\|\bx^{(k_2^\prime)}\| \right) \be_{k_2} \be_{k_2^\prime}^\top.
\end{align}
It then follows from \eqref{equ: second_order_mement_2} and \eqref{equ: second_order_monent_2_k} that
\begin{align}
\label{equ: second_order_moment_right}
& \quad \left\| \sum_{k_1 \ne k } \sum_{\psi_i = k_1 } \sum_{\psi_j = k}\EE(\bY_{k_1, i, j}^\top \bY_{k_1, i, j}) +   \sum_{\psi_i = \psi_j = k , i< j } \EE(\bY_{k, i, j}^{\top}\bY_{k, i, j})      \right\|\nonumber\\
& \le 4  ( \|\bx\|^2 + \|\bx^{(k)}\|^2 )n_k n s_n^2  \log n  \nonumber\\
& \quad \times \left\| \sum_{k_2=1}^K \sum_{k_2^\prime=1}^K \left(\|\bx^{(k_2)}\|_{\infty}+n^{-1/2}\|\bx^{(k_2)}\|\right)\left(\|\bx^{(k_2^\prime)}\|_{\infty}+n^{-1/2}\|\bx^{(k_2^\prime)}\| \right) \be_{k_2} \be_{k_2^\prime}^\top  \right\| \nonumber \\
& = 4  ( \|\bx\|^2 + \|\bx^{(k)}\|^2 )n_k n s_n^2  \log n  \sum_{k_2=1}^K \left(\|\bx^{(k_2)}\|_{\infty}+n^{-1/2}\|\bx^{(k_2)}\| \right)^2\nonumber \\
& \le 4( \|\bx\|^2 + \|\bx^{(k)}\|^2 ) \left\{ \left( \sum_{k_2=1}^K \|\bx^{(k_2)}\|_\infty^2\right)^{1/2}+n^{-1/2}\|\bx\| \right\}^2 n_k n s_n^2  \log n.
\end{align}
Combining \eqref{equ: second_order_moment_left} and \eqref{equ: second_order_moment_right}, under $E_1$ we obtain
\begin{align}
\label{equ: bound_varaince}
     & \quad \max\left\{ \sum_{\bY \in \bDelta} \bY\bY^{\top}, \sum_{\bY \in \bDelta} \bY^{\top}\bY  \right\} \nonumber\\
     & \le 4 ( \|\bx\|^2 + \|\bx^{(k)}\|^2 )\left\{ \left(\sum_{k_2=1}^K \|\bx^{(k_2)}\|_\infty^2\right)^{1/2}+n^{-1/2}\|\bx\|\right\}^2n_k n s_n^2  \log n.
\end{align}

\paragraph{Step 4. Concentration.} Finally, denoting the right hand side of (\ref{equ: bound_varaince}) as $\sigma_{\bDelta}^2$, by the matrix Bernstein inequality \citep[Theorem 1.6]{tropp2012user}, we have
 \[\quad \PP ( \|\tilde{\bI}_1\| > t| E_1)\le 2K \exp \left\{ -\frac{t^2}{2\sigma_{\bDelta}^2 + 2 R_{\bDelta}t/3}   \right\},\]
where $R_{\bDelta}$ and $\sigma_{\bDelta}^2$ are defined in \eqref{equ: R_y} and \eqref{equ: bound_varaince}, respectively.  Taking
\begin{equation*}
t = \frac{1}{2} \alpha(\delta) \sigma_{\bDelta} (\log n)^{1/2}, \text{ with } \alpha(\delta) = \frac{8\delta + 4(4\delta^2 + 9)^{1/2}}{3},
\end{equation*} 
we have
\begin{align*}
    & \quad \frac{2}{3}R_{\bDelta}t =  \frac{2}{3} \cdot 4\|\bx\|_{\infty}  \left\{ \left(\sum_{k_2=1}^K \|\bx^{(k_2)}\|_\infty^2\right)+n^{-1/2}\|\bx\|\right\}   (ns_n \log n)^{1/2} \cdot \frac{1}{2} \alpha(\delta) \sigma_{\bDelta}(\log n)^{1/2}\\
    & = \frac{2 \alpha(\delta)\|\bx\|_{\infty}\sigma_{\bDelta} (\log n)^{1/2}}{3 \left\{n_ks_n(\|\bx\|^2 + \|\bx^{(k)}\|^2) \right\}^{1/2}} \\
    & \quad \times 2 \left( \|\bx\|^2 + \|\bx^{(k)}\|^2 \right) \left\{ \left(\sum_{k_2=1}^K \|\bx^{(k_2)}\|_\infty^2\right)^{1/2}+n^{-1/2}\|\bx\|\right\} s_n (n_k n \log n )^{1/2} \\
    & = \frac{2 \alpha(\delta)\|\bx\|_{\infty} (\log n)^{1/2}}{3 \left\{n_ks_n(\|\bx\|^2 + \|\bx^{(k)}\|^2) \right\}^{1/2}} \sigma_{\bDelta}^2\\
    & \le \frac{2\delta\alpha(\delta)}{3} \sigma_{\bDelta}^2,
\end{align*}
where the last inequality comes from Assumption 5. As such, under the event $E_1$,  with probability at least $1 - 2K \exp\{- (3\alpha(\delta)^2 \log n)/(24 + 8\delta \alpha(\delta))\} = 1 - 2K/n^2$, we have 
\begin{equation}
\label{equ: I_1_tilde_upper_bound}
\|\tilde{\bI}_1\| \le \alpha(\delta) \left( \|\bx\|^2 + \|\bx^{(k)}\|^2 \right)^{1/2} \left\{ \left( \sum_{k_2=1}^K \|\bx^{(k_2)}\|_\infty^2\right)^{1/2}+ n^{-1/2}\|\bx\|\right\} s_n (n_k n)^{1/2} \log n. 
\end{equation}
Therefore, \eqref{equ: I_1_tilde_upper_bound} holds with probability at least $1 - 2K(n_k+1)/n^2$, which finishes the proof of \Cref{lem:Nk-tI1}. 
\end{proof}

\subsection{Proof of \Cref{lem:Nk-tS2}}
\begin{lemma}[Bounding $\tilde{\bS}_2$]\label{lem:Nk-tS2}
    Under the same assumptions as in \Cref{thm:consistency} with $\tilde{\bS}_2$ defined in \eqref{equ: S_2_tilde}, it holds that
    \begin{equation*}
     \|\tilde{\bS}_2\|\le \frac{4\times 6^{1/2}}{3}\max_{k^\prime \in [K]} \|\bx^{(k^\prime)}\|_4^2(n_k\log n)^{1/2}s_n+  4\|\bx\|_{\infty}^2 \log n.
    \end{equation*}
    with probability at least $1- 2K/n^2$.  
\end{lemma}
\begin{proof}[Proof of \Cref{lem:Nk-tS2}]
    Note that $\tilde{\bS}_2$ is a zero-mean diagonal matrix, and the diagonals are independent. We decompose $\tilde{\bS}_2$ as follows
    \begin{align*}
        \tilde{\bS}_2 & = \sum_{k_1=1, k_1 \ne k}^K  \sum_{\psi_i = k_1} \sum_{\psi_j = k} x_i^2  (\bA - \EE(\bA))_{i, j}  (\tilde{\bA} - \EE(\bA))_{i, j}    \be_{k_1} \be_{k_1}^\top \\
        & + \sum_{\psi_i = \psi_j =  k, i<j} (x_i^2 + x_j^2)  (\bA - \EE(\bA))_{i, j}  (\tilde{\bA} - \EE(\bA))_{i, j}  \be_{k} \be_{k}^\top. 
    \end{align*}
    Clearly, each diagonal of $\tilde{\bS}_2$ is now decomposed into a summation of independent random variables, and each random variable is a product of two independent and identically distributed centered Bernoulli random variables, up to some scaling coefficients. To investigate the properties of such random variables, we let $\xi$ and $\tilde{\xi}$ be two independent and identically distributed Bernoulli random variables with mean $p$. We have $\Var\big[(\xi - \EE(\xi))(\tilde{\xi} - \EE(\tilde{\xi}))\big] = (1-p)^2 p^2 \le \frac{4}{3}p^2$, and
    \begin{align*}
    \big\vert \EE\big[(\xi - \EE(\xi))(\tilde{\xi} - \EE(\tilde{\xi})) \big]^l \big\vert & = \big\vert  (1-p)^{2l}p^2 + 2[(-p)(1-p)]^lp(1-p) + p^{2l}(1-p)^2 \big\vert\\
    & = (1-p)^2 p^2 \big\vert  (1-p)^{2(l-1)} + 2[(-p)(1-p)]^{(l-1)} + p^{2(l-1)} \big\vert\\
    & \le \frac{l!}{2} \cdot \frac{4}{3}p^2 \cdot 1^{l-2}, 
    \end{align*}
    for $l\ge 3$, showing that $(\xi - \EE(\xi))(\tilde{\xi} - \EE(\tilde{\xi}))$ is a sub-Exponential random variable with parameter $(\frac{4}{3}p^2, 1)$ according to the Bernstein condition. It then follows that $x_i^2  (\bA - \EE(\bA))_{i, j}  (\tilde{\bA} - \EE(\bA))_{i, j}$ is a sub-Exponential random variable with parameter $(\frac{4}{3}x_i^4s_n^2, x_i^2)$ for $\psi_i\ne k$ while $\psi_j = k$, and $(x_i^2 + x_j^2)  (\bA - \EE(\bA))_{i, j}  (\tilde{\bA} - \EE(\bA))_{i, j}$ is a sub-Exponential random variable with parameter $\left(\frac{4}{3}(x_i^2 + x_j^2)^2 s_n^2, x_i^2 + x_j^2\right)$ for $\psi_i = \psi_j = k$. Moreover, $\sum_{\psi_i = k_1} \sum_{\psi_j = k} x_i^2  (\bA - \EE(\bA))_{i, j}  (\tilde{\bA} - \EE(\bA))_{i, j}$ is a sub-Exponential random variable with parameter $(\frac{4}{3} \|\bx^{(k_1)}\|_4^4n_ks_n^2, \|\bx^{(k_1)}\|_{\infty}^2)$ for $k_1 \ne k $, and $\sum_{\psi_i = \psi_j =  k, i<j}(x_i^2 + x_j^2)  (\bA - \EE(\bA))_{i, j}  (\tilde{\bA} - \EE(\bA))_{i, j}$ is sub-Exponential with parameter $\left(\frac{4}{3}\big[(n_k - 2) \|\bx^{(k)}\|_4^4 + \|\bx^{(k)}\|_2^4\big] s_n^2, 2\|\bx^{(k)}\|_{\infty}^2\right)$. By the Bernstein inequality, we have
    \begin{equation*}
        \PP \left( \left\vert \sum_{\psi_i = k_1} \sum_{\psi_j = k} x_i^2  (\bA - \EE(\bA))_{i, j}  (\tilde{\bA} - \EE(\bA))_{i, j} \right\vert > \left(\frac{8}{3} \|\bx^{(k_1)}\|_4^4n_ks_n^2 t \right)^{1/2} + \|\bx^{(k_1)}\|_{\infty}^2 t\right) \le 2e^{-t}, 
    \end{equation*}
    for any $t> 0 $. Taking $t = 2 \log n $, we have, with probability at least $1 - 2/n^2$, 
    \begin{align}
    \label{I_2_tilde_upper_bound_k_1}
   & \left\vert \sum_{\psi_i = k_1} \sum_{\psi_j = k} x_i^2  (\bA - \EE(\bA))_{i, j}  (\tilde{\bA} - \EE(\bA))_{i, j} \right\vert \le \frac{4\sqrt{3}}{3}\|\bx^{(k_1)}\|_4^2(n_ks_n^2 \log n)^{1/2} +2 \|\bx^{(k_1)}\|_{\infty}^2 \log n. 
    \end{align}
    Similarly, we have with probability at least $1-2/n^2$, 
    \begin{align}
    \label{I_2_tilde_upper_bound_k}
      & \quad \left\vert \sum_{\psi_i = \psi_j =  k, i<j} (x_i^2 + x_j^2)  (\bA - \EE(\bA))_{i, j}  (\tilde{\bA} - \EE(\bA))_{i, j}  \right\vert \nonumber \\
      & \le \left[ \frac{16}{3}\left\{(n_k - 2) \|\bx^{(k)}\|_4^4 + \|\bx^{(k)}\|_2^4\right\} s_n^2\log n\right]^{1/2}+  4\|\bx^{(k)}\|_{\infty}^2 \log n \nonumber\\
      &  \le \frac{4\sqrt{3}}{3} \left(n_k \|\bx^{(k)}\|_4^4 + \|\bx^{(k)}\|_2^4\right)^{1/2} \cdot (s_n^2\log n)^{1/2}+  4\|\bx^{(k)}\|_{\infty}^2 \log n \nonumber \\
      &  \le \frac{4 \sqrt{6}}{3} \|\bx^{(k)}\|_4^2 (n_k s_n^2\log n)^{1/2}+  4\|\bx^{(k)}\|_{\infty}^2 \log n. 
    \end{align}
    By the union bound, \eqref{I_2_tilde_upper_bound_k_1} holds for all $k_1\ne k$ and \eqref{I_2_tilde_upper_bound_k} holds simultaneously with probability at least $1- 2K/n^2$. Therefore, with probability at least $1 -2K/n^2$,
    \begin{align*}
    \|\tilde{\bS}_2\| & \le \frac{4\sqrt{3}}{3} \max \left\{2^{1/2} \|\bx^{(k)}\|_4^2, \max_{k_1 \ne k}\|\bx^{(k_1)}\|_4^2\right\}(n_k\log n)^{1/2}s_n \\
    & \quad+  2 \max\left\{2\|\bx^{(k)}\|_{\infty}^2, \max_{k_1 \ne k} \|\bx^{(k_1)}\|_{\infty}^2 \right\} \log n\\
    & \le \frac{4\sqrt{6}}{3} \max_{k^\prime \in [K]}\|\bx^{(k^\prime)}\|_4^2 (n_k\log n)^{1/2}s_n+  4\|\bx\|_{\infty}^2 \log n,
    \end{align*}
    which finishes the proof of \Cref{lem:Nk-tS2}.
\end{proof}

\subsection{Proof of \Cref{lem:Nk-S2}}
\begin{lemma}[Bounding $\bS_2 - {\EE(\bS_2)}$]\label{lem:Nk-S2}
    Under the same assumptions as in \Cref{thm:consistency}, it holds that
    \begin{equation*}
     \|\bS_2 - \EE(\bS_2)\|\le 2\times2^{1/2} \max_{k^\prime \in [K]}\|\bx^{(k^\prime)}\|_4^2(n_ks_n\log n)^{1/2}+  4\|\bx\|_{\infty}^2 \log n, 
    \end{equation*}
    with probability at least $1 - 2K/n^2$. 
\end{lemma}
\begin{proof}[Proof of \Cref{lem:Nk-S2}] Similar to the proof of \Cref{lem:Nk-tS2}, we decompose the diagonal matrix $\bS_2$ as follows
    \begin{align*}
        \bS_2 = \sum_{k_1 \ne k}  \sum_{\psi_i = k_1} \sum_{\psi_j = k} x_i^2  (\bA_{i, j} - \EE(\bA_{i, j}))^2 \be_{k_1} \be_{k_1}^\top+ \sum_{\psi_i = \psi_j =  k, i<j} (x_i^2 + x_j^2)  (\bA_{i, j} - \EE(\bA_{i, j}))^2  \be_{k} \be_{k}^\top. 
    \end{align*}
Clearly, each diagonal is now decomposed into a summation of a series of independent random variables. Since
\begin{align}
\label{equ: ES_2}
\EE(\bS_2) &= \sum_{k_1 \ne k} \sum_{\psi_i = k_1} \sum_{\psi_j = k} x_i^2  \EE(\bA_{i, j})(1- \EE(\bA_{i, j})) \be_{k_1} \be_{k_1}^\top \nonumber\\
& + \sum_{\psi_i = \psi_j =  k, i<j} (x_i^2 + x_j^2)  \EE(\bA_{i, j})(1- \EE(\bA_{i, j}))  \be_{k} \be_{k}^\top, 
\end{align}
we have 
\begin{align*}
\bS_2 - \EE(\bS_2) &  = \sum_{k_1 \ne k} \sum_{\psi_i = k_1} \sum_{\psi_j = k} x_i^2 \Big[ (\bA_{i, j} - \EE(\bA_{i, j}))^2 - \EE(\bA_{i, j})(1- \EE(\bA_{i, j}))  \Big] \be_{k_1} \be_{k_1}^\top\\
& + \sum_{\psi_i = \psi_j =  k, i<j} (x_i^2 + x_j^2) \Big[ (\bA_{i, j} - \EE(\bA_{i, j}))^2 - \EE(\bA_{i, j})(1- \EE(\bA_{i, j})) \big] \be_{k} \be_{k}^\top.  
\end{align*}
Let $\xi$ be a Bernoulli random variable with successful probability $p$. We now investigate the property of $(\xi - p)^2 - p(1-p)$.  Note that
\begin{align*}
   \big\vert \EE \big[ (\xi - p)^2 - p(1-p)\big]^l\big\vert &  =\big\vert \big[ (1 - p)^2 - p(1-p)\big]^lp + \big[p^2 - p(1-p)\big]^l(1-p) \big\vert\\
   & \le  p (1-p) |1-2p|^l \big[ (1-p)^{l-1} + p^{l-1}]\\
   & \le \begin{cases}
   p  \text{ if } l = 2, \\
   2p \text{ if } l \ge 3.
   \end{cases}
\end{align*}
That is,  $\big\vert \EE \big[ (\xi - p)^2 - p(1-p)\big]^l\big\vert \le  \frac{l!}{2} p 1^{l -2}$, for any $l\ge 2$. Such Bernstein condition yields that $(\xi - p)^2 - p(1-p)$ is a sub-Exponential random variable with parameter $(p, 1)$. It then follows that $x_i^2 \Big[ (\bA_{i, j} - \EE(\bA_{i, j}))^2 - \EE(\bA_{i, j})(1- \EE(\bA_{i, j}))\Big]$ is a sub-Exponential random variable with parameter $(x_i^4s_n, x_i^2)$, and $(x_i^2 + x_j^2) \Big[ (\bA_{i, j} - \EE(\bA_{i, j}))^2 - \EE(\bA_{i, j})(1- \EE(\bA_{i, j})) \big]$ is a sub-Exponential random variable with parameter $\left((x_i^2+x_j^2)^2s_n, x_i^2+ x_j^2\right)$. Moreover, $\sum_{\psi_i = k_1} \sum_{\psi_j = k} x_i^2 \Big[ (\bA_{i, j} - \EE(\bA_{i, j}))^2 - \EE(\bA_{i, j})(1- \EE(\bA_{i, j}))  \Big]$ is a sub-Exponential random variable with parameter $(\|\bx^{(k_1)}\|_4^4n_ks_n, \|\bx^{(k_1)}\|_{\infty}^2)$ for $k_1 \ne k $, and $ \sum_{\psi_i = \psi_j =  k, i<j} (x_i^2 + x_j^2)(x_i^2 + x_j^2) \Big[ (\bA_{i, j} - \EE(\bA_{i, j}))^2 - \EE(\bA_{i, j})(1- \EE(\bA_{i, j})) \big]$ is a sub-Exponential random variable with parameter $\left(\big[(n_k - 2) \|\bx^{(k)}\|_4^4 + \|\bx^{(k)}\|_2^4\big] s_n, 2\|\bx^{(k)}\|_{\infty}^2\right)$. By the Bernstein inequality and the union bound, we have, with probability at least $1 - 2K/n^2$,
\begin{align*}
& \quad \left \vert\sum_{\psi_i = k_1} \sum_{\psi_j = k} x_i^2 \Big[ (\bA_{i, j} - \EE(\bA_{i, j}))^2 - \EE(\bA_{i, j})(1- \EE(\bA_{i, j}))  \Big]\right\vert \\
& \le 2 \|\bx^{(k_1)}\|_4^2 (n_ks_n \log n)^{1/2} + 2\|\bx^{(k_1)}\|_{\infty}^2 \log n, 
\end{align*}
for any $k_1 \in [K]\backslash \{k\}$ and 
\begin{align*}
& \quad \left \vert      \sum_{\psi_i = \psi_j =  k, i<j} (x_i^2 + x_j^2) \Big[ (\bA_{i, j} - \EE(\bA_{i, j}))^2 - \EE(\bA_{i, j})(1- \EE(\bA_{i, j})) \big]      \right\vert \\
& \le 2  \left\{ \left((n_k - 2) \|\bx^{(k)}\|_4^4 + \|\bx^{(k)}\|_2^4\right)s_n \log n\right\}^{1/2} +  4\|\bx^{(k)}\|_{\infty}^2 \log n \\
& \le 2  \|\bx^{(k)}\|_4^2(2n_ks_n \log n)^{1/2} +  4\|\bx^{(k)}\|_{\infty}^2 \log n. 
\end{align*}
It immediately follows that, with probability at least $1 - 2K/n^2$,
    \begin{align*}
 \quad \|\bS_2- \EE(\bS_2)\|\le 2\times 2^{1/2} \max_{k^\prime \in [K]}\|\bx^{(k^\prime)}\|_4^2   (n_ks_n\log n)^{1/2}+  4\|\bx\|_{\infty}^2 \log n, 
    \end{align*}
which finishes the proof of \Cref{lem:Nk-S2}.
\end{proof}

\subsection{Proof of \Cref{lem:Nk-I2}}
\begin{lemma}[Bounding $\bI_2$]\label{lem:Nk-I2}
    Under the same assumptions as in \Cref{thm:consistency}, it holds that
    \begin{equation*}
     \|\bI_2\| \le \tilde{\alpha}(\delta) s_n \left\{( \|\bx\|^2 + \|\bx^{(k)}\|^2 )n_k s_n \sum_{k_2=1}^K \|\bx^{(k_2)}\|_1^2 \log n \right\}^{1/2},
    \end{equation*}
    with probability at least $1 -2K/n^2$, where $\tilde{\alpha}(\delta) = \{2\delta + 2(\delta^2 + 9)^{1/2}\}/3 < \alpha(\delta)/2$.
\end{lemma}
\begin{proof}[Proof of \Cref{lem:Nk-I2}]
Similar to the decomposition of $\bI_1$, we first decompose $\bI_2$ as follows
    \begin{align*}
    \bI_2 & = \bZ^\top \left[ \bx \bx^\top * (\bA-\EE(\bA)) \diag(\bZ_{., k}) \EE(\bA) \right]\bZ  \nonumber\\
    & = \sum_{k_1=1}^K \sum_{k_2 =1}^K  \sum_{\psi_i = k_1} \sum_{\psi_{i^\prime} = k_2} \sum_{\psi_j = k} x_i x_{i^\prime} (\bA - \EE(\bA) )_{i, j} \EE(\bA_{i^\prime, j}) \be_{k_1} \be_{k_2}^\top \nonumber\\
    & = \sum_{k_1 \ne k} \sum_{\psi_i = k_1} \sum_{\psi_j = k} (\bA - \EE(\bA) )_{i, j}\sum_{k_2 =1}^K  \sum_{\psi_{i^\prime} = k_2} x_i x_{i^\prime}  \EE(\bA_{i^\prime, j}) \be_{k_1} \be_{k_2}^\top  \nonumber\\
    & + \sum_{\psi_i = \psi_j = k, i<j} (\bA - \EE(\bA) )_{i, j}\sum_{k_2 =1}^K  \sum_{\psi_{i^\prime} = k_2} \big[x_i x_{i^\prime}  \EE(\bA_{i^\prime, j}]  + x_j x_{i^\prime}  \EE(\bA_{i^\prime, i})\big]\be_{k} \be_{k_2}^\top \nonumber\\
    & =  \sum_{k_1 \ne k} \sum_{\psi_i = k_1} \sum_{\psi_j = k} \bW_{k_1, i, j} +  \sum_{\psi_i = \psi_j = k, i<j}\bW_{k, i, j}, 
    \end{align*}
    where 
    \begin{align*}
    & \bW_{k_1, i, j} = (\bA - \EE(\bA) )_{i, j}\sum_{k_2 =1}^K  \sum_{\psi_{i^\prime} = k_2} x_i x_{i^\prime}  \EE(\bA_{i^\prime, j}) \be_{k_1} \be_{k_2}^\top, \text{ for } k_1 \ne k, \psi_i = k_1 \text{ and } \psi_j = k, \text{ and }\\
    & \bW_{k, i, j} = (\bA - \EE(\bA) )_{i, j}\sum_{k_2 =1}^K  \sum_{\psi_{i^\prime} = k_2} \big[x_i x_{i^\prime}  \EE(\bA_{i^\prime, j})  + x_j x_{i^\prime}  \EE(\bA_{i^\prime, i})\big]\be_{k} \be_{k_2}^\top, \text{ for } \psi_i = \psi_j = k. 
    \end{align*}
     Clearly,  $\tilde{\bDelta}_k = \{\bW_{k_1, i, j}: \psi_i=k_1, \psi_j = k, k_1 \ne k \} \cup \{\bW_{k, i, j}:\psi_i = \psi_j = k, i<j \}$ is a set of zero-mean independent random matrices. Moreover, 
     \begin{align*}
      \left\|\bW_{k_1, i, j}\right \| &\le  = \left\{\sum_{k_2 =1}^K  \left(\sum_{\psi_{i^\prime} = k_2} x_i x_{i^\prime}  \EE(\bA_{i^\prime, j}) \right)^2 \right\}^{1/2}\le \|\bx\|_{\infty}s_n \left(\sum_{k_2 = 1}^K \|\bx^{(k_2)}\|_1^2\right)^{1/2}, 
     \end{align*}
     for $k_1 \ne k$, and 
         \begin{align*}
      \left\|\bW_{k, i, j}\right \| &\le \left \|\sum_{k_2 =1}^K  \sum_{\psi_{i^\prime} = k_2} \big[x_i x_{i^\prime}  \EE(\bA_{i^\prime, j})  + x_j x_{i^\prime}  \EE(\bA_{i^\prime, i})\big]\be_{k} \be_{k_2}^\top\right\| \\
      & = \left\{\sum_{k_2 =1}^K  \left(\sum_{\psi_{i^\prime} = k_2}x_i x_{i^\prime}  \EE(\bA_{i^\prime, j})  + \sum_{\psi_{i^\prime = k_2}}x_j x_{i^\prime}  \EE(\bA_{i^\prime, i})\right)^2 \right\}^{1/2}\\
      &\le 2\|\bx\|_{\infty}s_n \left(\sum_{k_2 = 1}^K \|\bx^{(k_2)}\|_1^2\right)^{1/2}, 
     \end{align*} 
     As such, define
\begin{equation}
\label{equ: R_w}
    R_{\tilde{\bDelta}}= 2\|\bx\|_{\infty}s_n \left(\sum_{k_2 = 1}^K \|\bx^{(k_2)}\|_1^2\right)^{1/2}, 
\end{equation}
which can serve as a uniform upper bound for the spectral norms of the matrices in $\tilde{\bDelta}$. We next turn to upper bound the second order moments for the matrices in $\tilde{\bDelta}$. For $k_1\ne k$,
\begin{align*}
\EE(\bW_{k_1, i, j}\bW_{k_1, i, j}^{\top}) &= \EE((\bA - \EE(\bA) )_{i, j}^2) \sum_{k_2 =1}^K \sum_{\psi_{i^\prime} = k_2} \sum_{\psi_{i^{\prime\prime}} = k_2}  x_i^2 x_{i^\prime}x_{i^{\prime \prime}}\EE(\bA_{i^\prime, j})\EE(\bA_{i^{\prime\prime}, j}) \be_{k_1} \be_{k_1}^\top\\
& \preceq s_n^3 x_i^2  \sum_{k_2=1}^K \left( \sum_{\psi_{i^\prime} = k_2} x_{i^\prime} \right)^2\be_{k_1} \be_{k_1}^\top \preceq s_n^3 x_i^2  \sum_{k_2=1}^K  \|\bx^{(k_2)}\|_1^2\be_{k_1} \be_{k_1}^\top.
\end{align*}
 It then follows that 
\begin{align}
 \sum_{k_1\ne k}\sum_{\psi_i = k_1} \sum_{\psi_j = k}\EE(\bW_{k_1, i, j}\bW_{k_1, i, j}^{\top}) \notag & \preceq  n_k s_n^3  \sum_{k_2=1}^K  \|\bx^{(k_2)}\|_1^2 \sum_{k_1\ne k} \|\bx^{(k_1)}\|^2  \be_{k_1} \be_{k_1}^\top \nonumber \\
    &  \preceq  n_k s_n^3  \sum_{k_2=1}^K  \|\bx^{(k_2)}\|_1^2 \max_{k^\prime \in [K]\backslash\{k\}} \|\bx^{(k^\prime)}\|^2\sum_{k_1\ne k}  \be_{k_1} \be_{k_1}^\top.   \label{equ: WWT_1}
\end{align}
In addition, for any $i<j$ such that $\psi_i = \psi_j = k$, we have
\begin{align*}
\EE(\bW_{k, i, j}\bW_{k, i, j}^{\top}) \preceq 2s_n^3  \sum_{k_2=1}^K \left[ x_i^2  \|\bx^{(k_2)}\|_1^2 + x_j^2\|\bx^{(k_2)}\|_1^2\right] \be_k \be_k^\top \preceq  2s_n^3 (x_i^2 + x_j^2)  \sum_{k_2=1}^K \|\bx^{(k_2)}\|_1^2 \be_k \be_k^\top.
\end{align*}
Taking the summation for all $i<j$, we have
\begin{equation}
\label{equ: WWT_2}
\sum_{\psi_i = \psi_j = k, i<j}\EE(\bW_{k, i, j}\bW_{k, i, j}^{\top}) \preceq 2s_n^3 (n_k - 1) \|\bx^{(k_2)}\|^2 \sum_{k_2=1}^K \|\bx^{(k)}\|_1^2 \be_k \be_k^\top.
\end{equation}
Combing \eqref{equ: WWT_1} and \eqref{equ: WWT_2} yields that
\begin{align}
\label{equ: W_som_1}
&\quad \left\|\sum_{\bW \in \tilde{\bDelta}}\EE(\bW\bW^{\top}) \right\| \le 2 n_k s_n^3 \max_{k^\prime \in [K]} \|\bx^{(k^\prime)}\|^2 \sum_{k_2=1}^K \|\bx^{(k)}\|_1^2. 
\end{align}
Similarly, for $k_1 \ne k$, we have
\begin{align*}
\EE(\bW_{k_1, i, j}^\top \bW_{k_1, i, j}) & \preceq s_n^3 x_i^2  \sum_{k_2=1}^K  \sum_{k_2^\prime =1}^K \|\bx^{(k_2)}\|_1 \|\bx^{(k_2^\prime)}\|_1 \be_{k_2} \be_{k_2^\prime}^\top, 
\end{align*}
which leads to 
\begin{align}
\label{equ: WTW_1}
& \quad \sum_{k_1 \ne k } \sum_{\psi_i = k_1 } \sum_{\psi_j = k}\EE(\bW_{k_1, i, j}^\top \bW_{k_1, i, j}) \preceq  (\|\bx\|^2 - \|\bx^{(k)}\|^2)n_k s_n^3 \sum_{k_2=1}^K  \sum_{k_2^\prime =1}^K \|\bx^{(k_2)}\|_1 \|\bx^{(k_2^\prime)}\|_1 \be_{k_2} \be_{k_2^\prime}^\top.  
\end{align}
For $i< j$ with $\psi_i = \psi_j = k$, 
\begin{align*}
\EE(\bW_{k, i, j}^{\top} \bW_{k, i, j}) & \preceq s_n^3  \sum_{k_2=1}^K \sum_{k_2^\prime=1}^K  \left( |x_i| \| \bx^{(k_2)}\|_1 + |x_j|\| \bx^{(k_2)}\|_1\right) \left(|x_i|   \|\bx^{(k_2^\prime)}\|_1 + |x_j| \|\bx^{(k_2^\prime)}\|_1 \right) \be_{k_2} \be_{k_2^\prime}^\top\\
& \preceq  2 (x_i^2 + x_j^2)  s_n^3 \sum_{k_2=1}^K \sum_{k_2^\prime=1}^K  \|\bx^{(k_2)}\|_1 \|\bx^{(k_2^\prime)}\|_1\be_{k_2} \be_{k_2^\prime}^\top,
\end{align*}
leading to 
\begin{align}
\label{equ: WTW_2}
& \quad \sum_{\psi_i = \psi_j = k , i< j } \EE(\bW_{k, i, j}^{\top}\bW_{k, i, j}) \preceq  2 \|\bx^{(k)}\|^2(n_k-1) s_n^3 \sum_{k_2=1}^K \sum_{k_2^\prime=1}^K  \|\bx^{(k_2)}\|_1 \|\bx^{(k_2^\prime)}\|_1\be_{k_2} \be_{k_2^\prime}^\top. 
\end{align}
It then follows from \eqref{equ: WTW_1} and \eqref{equ: WTW_2} that
\begin{align}
\label{equ: W_som_2} 
\left\| \sum_{\bW \in \tilde{\bDelta}}\EE(\bW^\top \bW) \right\| & \le 2 ( \|\bx\|^2 + \|\bx^{(k)}\|^2 )n_k s_n^3  \left\| \sum_{k_2=1}^K \sum_{k_2^\prime=1}^K  \|\bx^{(k_2)}\|_1 \|\bx^{(k_2^\prime)}\|_1\be_{k_2} \be_{k_2^\prime}^\top \right\| \nonumber\\
& = 2  ( \|\bx\|^2 + \|\bx^{(k)}\|^2 )n_k s_n^3 \sum_{k_2=1}^K \|\bx^{(k_2)}\|_1^2. 
\end{align}
Combining \eqref{equ: W_som_1} and \eqref{equ: W_som_2} yields that 
\begin{align}
\label{equ: W_bound_varaince}
\sigma_{\tilde{\bDelta}}^2 := \max\left\{ \left\|\sum_{\bW \in \tilde{\bDelta}} \bW\bW^{\top}\right\|, \left\|\sum_{\bW \in \tilde{\bDelta}} \bW^{\top}\bW \right\| \right\}  \le 2 ( \|\bx\|^2 + \|\bx^{(k)}\|^2 )n_k s_n^3 \sum_{k_2=1}^K \|\bx^{(k_2)}\|_1^2. 
\end{align}
 Finally, by the matrix Bernstein inequality,  we have
 \[\quad \PP ( \|\bI_2\| > t)\le 2K \exp \left\{ -\frac{t^2}{2\sigma_{\tilde{\bDelta}}^2 + 2 R_{\tilde{\bDelta}}t/3}   \right\},\]
where $R_{\tilde{\bDelta}}$ and $\sigma_{\tilde{\bDelta}}^2$ are defined in \eqref{equ: R_w} and \eqref{equ: W_bound_varaince}, respectively.  Taking
\begin{equation*}
t =  \frac{1}{2^{1/2}}\tilde{\alpha}(\delta) \sigma_{\bDelta} (\log n)^{1/2}, \text{ with } \tilde{\alpha}(\delta) = \frac{2\delta + 2(\delta^2 + 9)^{1/2}}{3} < \frac{1}{2}{\alpha(\delta)}, 
\end{equation*} 
we have
\begin{align*}
& \quad \frac{2}{3}R_{\tilde{\bDelta}}t =  \frac{2}{3} \cdot 2\|\bx\|_{\infty}s_n (\sum_{k_2 = 1}^K \|\bx^{(k_2)}\|_1^2)^{1/2} \cdot  \frac{1}{2^{1/2}}\tilde{\alpha}(\delta) \sigma_{\bDelta} (\log n)^{1/2}\\
& = \frac{2 \tilde{\alpha}(\delta)\|\bx\|_{\infty}\sigma_{\tilde{\bDelta}}(\log n)^{1/2}}{3 (n_ks_n(\|\bx\|^2 + \|\bx^{(k)}\|^2}) )^{1/2}  \cdot 2^{1/2} \cdot ( \|\bx\|^2 + \|\bx^{(k)}\|^2 )^{1/2} \cdot (n_k s_n^3)^{1/2} \cdot (\sum_{k_2=1}^K \|\bx^{(k_2)}\|_1^2)^{1/2}\\
& \le \frac{2\delta \tilde{\alpha}(\delta)}{3}  \sigma_{\tilde{\bDelta}}^2, 
\end{align*}
where the last inequality comes from \Cref{asm:community-sizes}. Hence, with probability at least $1 - 2K\exp\{- \tilde{\alpha}(\delta)^2 \log n/(2 + 2\delta \tilde{\alpha}(\delta)/3)\} = 1 - 2K/n^2$, we have 
\begin{equation}
\label{equ: I_2_upper_bound}
\|\bI_2\| \le \tilde{\alpha}(\delta) s_n \left\{( \|\bx\|^2 + \|\bx^{(k)}\|^2 )n_k s_n \sum_{k_2=1}^K \|\bx^{(k_2)}\|_1^2 \log n \right\}^{1/2},
\end{equation}
which finishes the proof of \Cref{lem:Nk-I2}.
\end{proof}

\putbib[references]
\end{bibunit}

\end{document}